\documentclass{article}

\usepackage{arxiv}

\usepackage[utf8]{inputenc} 
\usepackage[T1]{fontenc}    
\usepackage[hidelinks]{hyperref}       
\usepackage{booktabs}       
\usepackage{amsfonts}       
\usepackage{nicefrac}       
\usepackage{microtype}      
\usepackage{graphicx}
\usepackage{doi}

\usepackage{booktabs} 

\usepackage{amssymb,amsthm,amsfonts,mathrsfs}

\usepackage{mathtools} 

\usepackage{enumitem}
\usepackage{graphicx}
\usepackage{subcaption}
\usepackage{balance}

\usepackage{xcolor,colortbl}
\usepackage{algpseudocode,algorithm,algorithmicx}
\newcommand*\Let[2]{\State #1 $\gets$ #2}
\algrenewcommand\algorithmicrequire{\textbf{VARs:}}

\newtheorem{theorem}{Theorem}
\newtheorem{invariant}{Invariant}

\usepackage{tikz}
\newcommand*\circled[1]{\tikz[baseline=(char.base)]{
		\node[shape=circle,draw,inner sep=1pt] (char) {#1};}}

\usepackage{biblatex}
\newcommand{\preprint}{1}

\title{Highly Available Queue-oriented Speculative Transaction Processing}

\author{\href{https://orcid.org/0000-0003-0754-0504}{\includegraphics[scale=0.06]{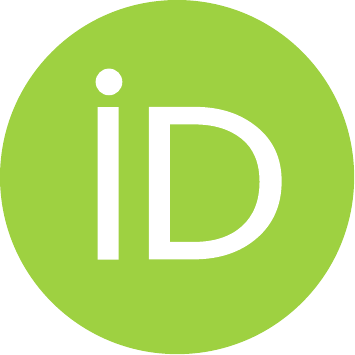}\hspace{1mm}Thamir~M.~Qadah$^1$}~
	and~\href{https://orcid.org/0000-0003-2779-6080}{\includegraphics[scale=0.06]{orcid.pdf}\hspace{1mm}Mohammad~Sadoghi$^2$}
		\\
		Exploratory Systems Lab \\
		$^1$ Computer Science Department, Umm Al-Qura University\\
		$^2$ Computer Science Department, UC Davis\\
		$^1$\texttt{tmqadah@uqu.edu.sa},
	 	$^2$\texttt{msadoghi@ucdavis.edu}
}

\newcommand{\ackmsg}{{\small \sf ACK}}

\newcommand{\quecc}{{\small \sf QR-Store}}
\newcommand{\qcmw}{{\small \sf QueCC}}

\newcommand{\calvin}{{\small \sf Calvin}}

\newcommand{\qstore}{{\small \sf Q-Store}}
\newcommand{\qszero}{{\small \sf Q-Store}}
\newcommand{\qsone}{{\small \sf QR-Store-rf1}}
\newcommand{\qstwo}{{\small \sf QR-Store-rf2}}

\newcommand{\calvinone}{{\small \sf Calvin-rf1}}
\newcommand{\calvintwo}{{\small \sf Calvin-rf2}}

\newcommand{\qsonenode}{{\small \sf QueCC-R}}

\newcommand{\qsfournode}{{\small \sf QR-Store}}

\newcommand{\raft}{{\small \sf RAFT}}
\newcommand{\zab}{{\small \sf ZAB}}

\newcommand{\eat}[1]{} 

\newcommand{\plotwidth}[2]{\ifdefined\preprint
	#2
	\else
	#1
	\fi}

\begin{document}
\maketitle

\begin{abstract}
Deterministic database systems have received increasing attention from the database research community in recent years. Despite their current limitations, recent proposals of distributed deterministic transaction processing systems demonstrated significant improvements over systems using traditional transaction processing techniques (e.g., two-phase-locking or optimistic concurrency control with two-phase-commit). However, the problem of ensuring high availability in deterministic distributed transaction processing systems has received less attention from the research community, and this aspect has not been analyzed and evaluated well. This paper proposes a generic framework to model the replication process in deterministic transaction processing systems and use it to study three cases. We design and implement QR-Store, a queue-oriented replicated transaction processing system, and extensively evaluate it with various workloads based on a transactional version of YCSB. Our prototype implementation QR-Store can achieve a throughput of 1.9 million replicated transactions per second in under 200 milliseconds and a replication overhead of 8\%-25\% compared to non-replicated configurations. 
\end{abstract}

\keywords{database systems \and deterministic transaction processing\and distributed transaction management\and performance evaluation \and high availability}

\section{Introduction
\label{sec:introduction}}
Cloud providers continue to provide a virtual computing infrastructure that provides a higher amount of main memory and virtual CPU cores. Currently, for instance, Amazon Web Services provides virtual instances configurations that are equipped with up to $448$ vCPUs, $24 TB$ of memory, and $100 Gbps$ network connectivity.\footnote{https://aws.amazon.com/ec2/instance-types/high-memory/}
Therefore, there is a growing demand for utilizing this modern computing infrastructure efficiently. 

Many deterministic database systems are proposed in the research literature to utilize modern computing infrastructures more efficiently (e.g., \cite{kallman_h-store_2008, thomson_calvin_2012, abadi_overview_2018}). Recent proposals of distributed deterministic transaction processing (DTP) systems demonstrated significant improvements over systems using traditional transaction processing techniques (e.g., 2PL/OCC+2PC). While distributed DTP systems have shown significant improvements in transaction processing performance, many database applications require high availability. For example, users of online banking applications desire that it is available $24\times7$. Furthermore, cloud providers' service level agreements promise at least four nines  (i.e., $99.99\%$ availability).\footnote{https://azure.microsoft.com/en-us/support/legal/sla/mysql} Database replication for traditional transaction processing protocols is well-studied (e.g., \cite{ozsu_principles_2020, kemme_database_2010, kemme_database_2010-1}). In contrast, the problem of ensuring high availability via replication in distributed DTP systems has received less attention from the research community, and this aspect of distributed DTP systems has not been analyzed and evaluated well. 

We consider database systems where the database state can be partitioned and distributed across multiple nodes (e.g., \cite{kallman_h-store_2008, thomson_calvin_2012, faleiro_lazy_2014, faleiro_rethinking_2015, lu_star_2019, lu_aria_2020, zamanian_chiller_2020}). Furthermore, the partitioned database state is replicated for high availability. With deterministic transaction processing, the replication is simplified because transaction histories are deterministic and strictly serializable. Strict serializability implies that transaction execution of conflicting transactions follows a {\em single} order across all partitions and replicas. By requiring that the predetermined order is followed during execution and in the replicated state, the replication process is simplified because the new database state is guaranteed to be equivalent due to deterministic execution. Thus, the key challenge with replication in distributed DTP systems is the negative impact of performing replication on the transaction processing performance.  

To address this challenge, we build on our highly efficient queue-oriented transaction processing paradigm \cite{qadah_quecc:_2018, qadah_queue-oriented_2019, qadah_q-store_2020}. In our earlier work \qcmw{}\cite{qadah_quecc:_2018}, we addressed the issue of the overhead of exiting concurrency control techniques under high-contention workloads and demonstrate that speculative and queue-oriented transaction processing can improve the system's throughput by up to two orders of magnitudes over the state-of-the-art centralized (non-partitioned) transaction processing systems. In \qstore{}\cite{qadah_q-store_2020}, we improve the efficiency of distributed and partitioned DTP systems by employing queue-oriented transaction processing techniques and demonstrate up to $22\times$ better performance. 

In this paper, we propose a generalized framework to analyze the design space of distributed and replicated deterministic transaction processing systems and extend \qcmw{} and \qstore{} with replication support for high availability. Based on the proposed framework, we propose a primary-copy approach and perform eager, speculative, queue-oriented replication to mitigate the overhead of replication in distributed DTP systems. Our approach amortizes the cost of replication and transaction processing over batches of transactions and processes these batches in parallel on a replicated clusters of server nodes. Furthermore, we exploit the fact that deterministic transaction execution and replication in DTP systems are independent, which allows us to either fully or partially hide the cost of replication while ensuring safe and strictly serializable transaction execution. 

Our contributions in this paper can be summarized as follows: 

\begin{itemize}
	\item we propose a generalized framework for DTP systems, a unified replication API for DTP systems, and apply the framework on three systems from the literature (Section \ref{sec:genfw});
	
	\item we design \quecc{}, a highly available queue-oriented and replicated transaction processing system version of \qstore{} (Section \ref{sec:qrep-tp});
	
	\item we prototype \quecc{} and propose optimization techniques to improve the performance of state-of-the-art in queue-oriented deterministic transaction processing (Section \ref{sec:qrep-impl});
	
	\item we extensively evaluate \quecc{} using standard benchmarks such as YCSB (Section \ref{sec:qrep-eval}).
\end{itemize}


\begin{figure}[!t]
	\centering
	\includegraphics[width=\plotwidth{0.98}{0.7}\linewidth,trim=0cm 0cm 0cm 0cm,clip]{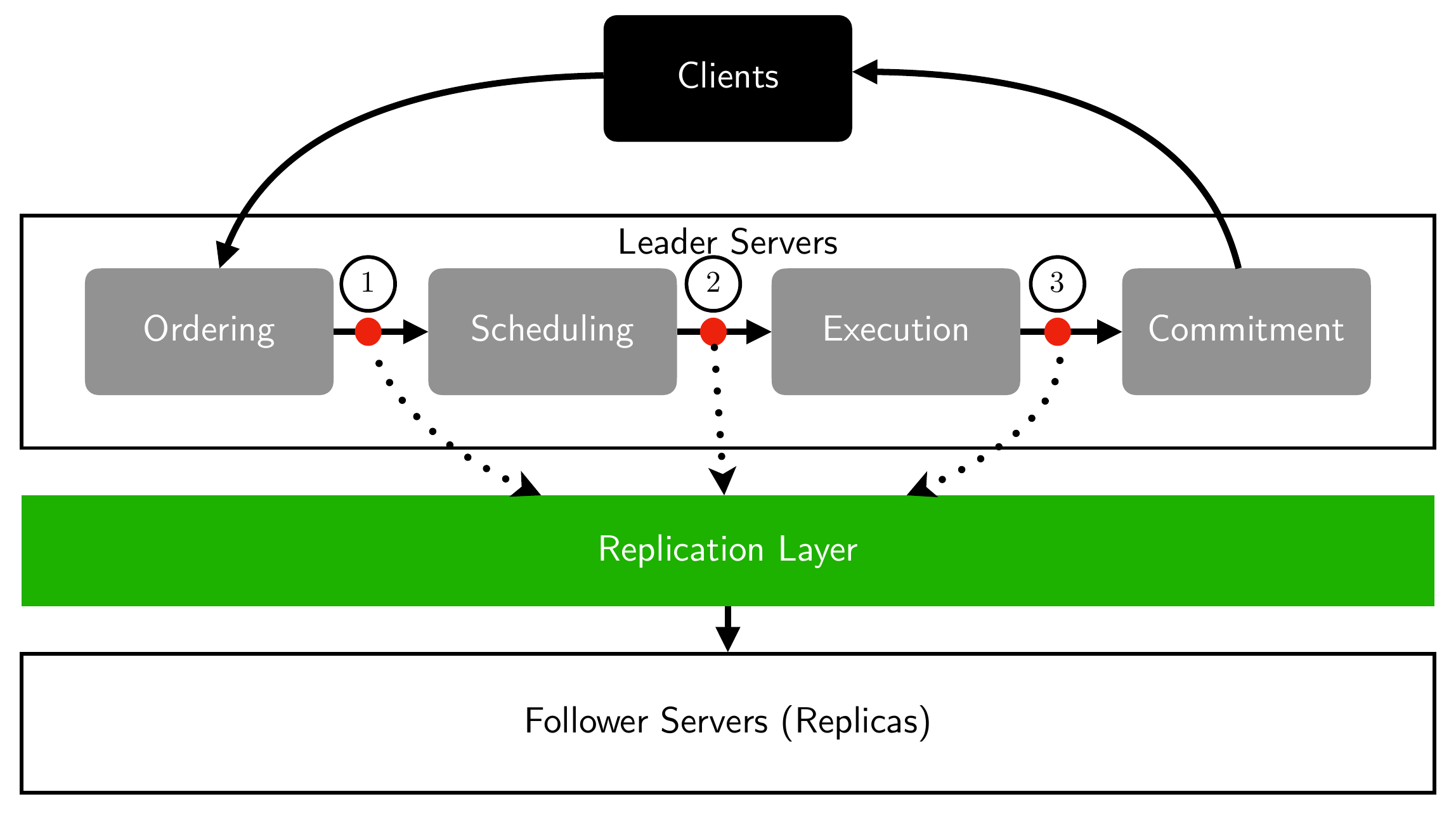}
	\caption[]{Generalized Deterministic Transaction Processing Framework}
	\label{fig:general-fw}
	\vspace{-4mm}
\end{figure}

\section{A Generalized Replication Framework}
\label{sec:genfw}
We propose a generalized DTP framework. In Figure \ref{fig:general-fw}, we illustrate a simple framework that adopts the client-server architecture. The system is composed of a set of clients that sends transactions for processing to a set of servers. Clients receive commitment responses from servers when their submitted transactions are committed to the database. The transaction processing workflow by a leader set of servers is composed of four generic stages for processing transactions deterministically with strict serializability. These steps are {\em ordering}, {\em scheduling}, {\em execution}, and {\em commitment}. The processing work in each stage can be done in a parallel and distributed fashion to improve the system's performance (e.g., by multiple worker threads deployed on a set of machines). It is important to note that DTP systems use batching to improve the throughput performance of the system.

A DTP system ensures strict serializability of transaction histories by predetermining the {\em order} of transactions execution/commitment before scheduling them for execution. In the {\em scheduling} stage, the scheduling algorithm needs to guarantee that the execution and the commitment of the transactions do not violate the predetermined order. In between stages, we define replication points (red circles in Figure \ref{fig:general-fw}). These are points in the transaction processing workflow where replication can be done. The output of a stage can be replicated using the replication layer to achieve system high availability. A set of follower servers (replicas) get the replicated output from a stage and proceed to use it as an input to the next stage.
 


In our framework, the replication layer is a logical layer. One implementation approach is by using a shim that interacts with a replicated coordination service such as Zookeeper \cite{apache_zookeeper_website} and etcd \cite{etcd_website} or publish-subscribe systems like Kafka\cite{apache_kafka_website}. In this approach, the service serves as a physical middleware between the leader set of servers and the follower set of servers. Existing work uses Zookeeper in their prototype implementations (e.g., \cite{thomson_calvin_2012}), but Zookeeper is not designed for this purpose. In our experiments, we observed that Zookeeper could not handle the replication load when the replication request rate is high. Therefore, using a service like Kafka appears to be a better option, and we plan to study that in future work.

Another implementation approach for the the replication layer is having the shim implements a protocol such as Paxos \cite{lamport_paxos_2002},  Viewstamped replication \cite{oki_viewstamped_1988}, or Raft\cite{ongaro_search_2014} directly. This integrated approach has a lower overhead (no need for additional dedicated nodes for the replication layer). However, it involves a more tight integration with the DTP system and is more complex to realize. 

To realize both approaches in a generalized way, we introduce a simple API that abstracts away the complexity and hides the details behind the underlying implementations. The API is compromised of two simple functions \textproc{replicateData} and \textproc{receiveData}. More details about this API are presented in Section \ref{sec:qrep-impl}.


Our proposed framework is general enough to admit existing work on deterministic transaction processing systems as specialized implementations. We discuss three case studies to illustrate the applicability of the proposed generalized framework to provide a unified framework to understand DTP systems.

\subsection{Case Study: Calvin}
\label{sec:case-studies-calvin}
\calvin{} \cite{thomson_calvin_2012} is one of the first DTP systems that supports replication. In \calvin{}, the {\em ordering} stage performs epoch-based batching of transactions and it is called {\em sequencing}. \calvin{} uses the replication point that follows the {\em ordering} stage and replicates batches of sequenced transactions. In the {\em scheduling} stage, \calvin{} uses a deterministic locking algorithm to schedule transactions for execution. In deterministic locking, the order of lock acquisition follows the predetermined transaction order by the sequencing layer. \calvin{} is a distributed DTP system and requires the use of a distributed commit protocol (CP) in the {\em commitment} stage. However, it avoids using the traditional heavyweight two-phase commit protocol and uses a lightweight CP that exploits deterministic execution. Transactions in \calvin{} commit when all the operations of the distributed transactions complete. The CP aborts transactions when the transaction has a logic-induced abort, and it is aborted deterministically across all partitions and replicas. In the absence of a logic-induced abort, transactions are committed, and the commit response is sent to the clients by the sequencing node that originally received the transaction and sequenced it. 

\calvin{}'s original proposal \cite{thomson_calvin_2012} proposed replicating the output of the {\em ordering} stage. However, based on our framework, it is possible to use other replication points. For example, the updated records in the execution stage can be logged and replicated before commitment.

\subsection{Case Study: Q-Store} 
\label{sec:case-studies-qs}
\qstore{} \cite{qadah_q-store_2020} is also another distributed DTP system, but unlike \calvin{}, it combines the {\em ordering} stage and the {\em scheduling} stage into a single parallel stage called {\em planning}. The {\em execution} stage uses the concept of execution queues (EQs) of operations as an execution primitive while \calvin{} uses the concept of a transaction as an execution primitive. \qstore{} focuses on distributed transaction processing without replication. The {\em planning} stage maps batches of transactions to execution queues tagged with execution priorities. The execution stage executes them based on their priorities, and the commitment stage maps them back to transactions and sends responses to clients. 

Replication in \qstore{} can use any one of the replication points. The first replication point occurs in the middle of planning and is similar to the replication in \calvin{} (i.e., replicating sequence batches of transactions). At follower nodes, the replicated sequence is planned into execution queues. Interestingly, if the first replication point is used, it is possible to have heterogeneous configurations of servers. For example, a group of servers can follow \calvin{}'s DTP approach while the others can follow \qstore{}'s.

When using the second replication point, which is after the {\em planning} stage, a novel replication scheme emerges. Because the execution primitive of \qstore{} is a set of execution-queues (a.k.a. plans), the set is replicated to follower servers, and the follower servers can take the replicated plans and use them in the next stages. With this approach, \qstore{} is also required to replicate transaction contexts so that in the commitment stage, the replicas can map execution queues back to transactions for commitment. 

\qstore{} can also use the third replication point, which also introduces yet another novel replication scheme. In this case, instead of creating traditional logs, \qstore{} creates plans of execution queues containing write-only operations of updated records. When replicated successfully, it is fed to the execution stage at the replicas, and no specialized stage is needed to process the replicated plans. Furthermore, only the last write operation on the record needs to be inserted in the write-only execution queues.

\subsection{Case Study: QueCC} 
\label{sec:case-studies-qc}
\qcmw{} \cite{qadah_quecc:_2018} is a single node DTP system that is designed and optimized for multi-socket, many-core machines. \qcmw{} uses the same concepts as \qstore{} in terms of having {\em planning} and {\em execution} stages, but all stages are parallel but not distributed by design. \qcmw{} can be extended to become a replicated DTPS. In this case, the leader server set contains only a single node that contains the entire database state, and its state is replicated using the replication layer. Compared to \qstore{}, \qcmw{} does not exploit partitioning and horizontal scalability; however, it can scale vertically by using more cores. Furthermore, it is possible to have heterogeneous hardware configurations for replicas where replicas don't have the same hardware specifications as the leader node.


\begin{figure*}[!t]
	\includegraphics[width=\plotwidth{0.8}{0.98}\textwidth,trim=0cm 0cm 0cm 0cm,clip]{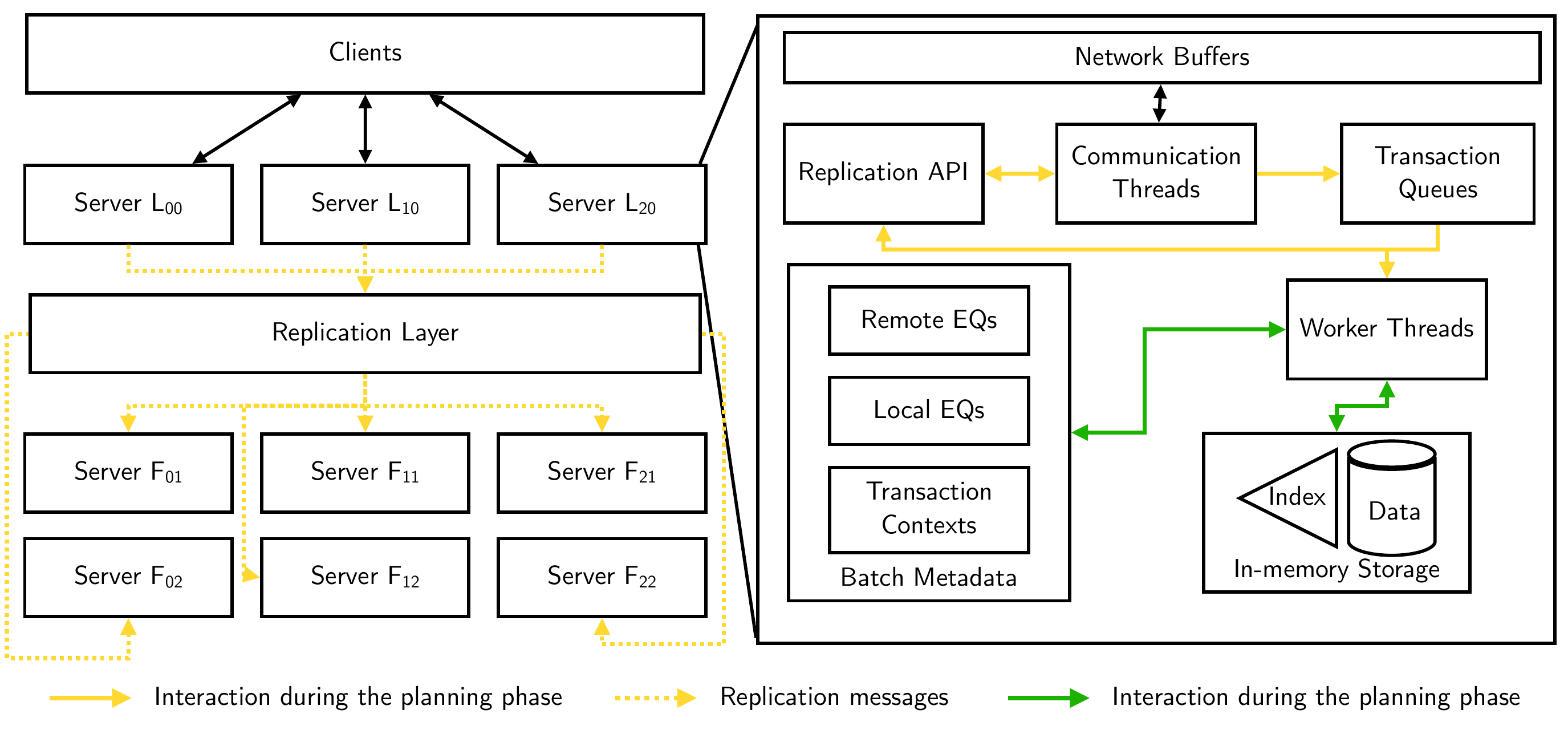}
	\centering
	\vspace{-3mm}
	\caption{System Architecture}
	\label{fig:sys-arch-r}
	\vspace{-4mm}
\end{figure*} 

\section{Highly Available Queue-oriented Transaction Processing}
\label{sec:qrep-tp}
Based on the generalized framework described in Section \ref{sec:genfw}, we focus on designing replication schemes for \qstore{} and study their impact on the system's performance. We build \quecc{} which is a replicated version of \qstore{} and give some overview of \quecc{}.

\subsection{QR-Store's Architecture}

As a distributed DTP system, \quecc{} runs on a cluster of nodes. Each node holds a partition of the database. It supports processing multi-partition transactions where a transaction may access records from different partitions. Each partition is replicated independently with a replication factor $rf$. For example, if $rf=2$ for partition $p_0$, then the system has three nodes hosting $p_0$. One of them is the leader node, while the others are followers. On the left side of Figure \ref{fig:sys-arch-r}, we show an example system architecture with three partitions and a replication factor of three (i.e., $rf=2$). Visually, horizontal grouping of nodes implies a cluster of \quecc{} nodes comprising a full replica of the distributed database instance, while vertical grouping implies replication groups. For example, nodes $L_{00}$, $L_{10}$, and $L_{20}$ form a cluster instance of \quecc{}, while nodes $L_{00}$, $F_{01}$, and $F_{02}$ form a replication group. Note that replication messages are communicated within a replication group only, while all other messages related to processing transactions are communicated within a cluster instance.  This communication scheme ensured minimal communication among nodes in \quecc{}.

On the right side of Figure \ref{fig:sys-arch-r}, we show the key components internal to a server node. Each node receives client transactions that are processed by communication threads into a set of client {\em Transaction Queues}. {\em Worker threads} on each node process client transactions in two phases: {\em planning} and {\em execution}. Note that we consider the commitment stage as part of the execution phase. In the planning phase, worker threads create execution queues (EQs) that access a sub-partition of the node's partition. To facilitate scheduling of EQs during the execution phase in \quecc{}, each worker thread in \quecc{} tags its EQs with a priority value. This value can be static or dynamic, but we assume statically predetermined priorities. There are {\em remote} EQs and {\em local} EQs. Remote EQs are executed at remote nodes as transaction fragments in them access other remote partitions. In addition to EQs, transaction contexts are maintained, which captures transaction dependencies and other transactions metadata. EQs and transaction contexts are stored in the {\em Batch metadata}, which distributed shared data structure. Furthermore, worker threads in the leader set of nodes use the {\em Replication API} to facilitate replication of the Batch metadata to the replicas. 

During the execution phase, worker threads execute and commit EQs based on their priorities. For example, say we have two EQs $q_i$ and $q_j$. The fact that $q_i$ can be either remote or local is orthogonal, and the same applies to $q_j$. Let $pr(q)$ denote the priority of an EQ $q$. 
\quecc{} maintains a global execution invariant such that $q_i$ is executed before $q_j$ {\em if and only if} $pr(q_i) > pr(q_j)$. Maintaining this global execution invariant with a cluster ensures a single global order of conflicting operations, which produce strict serializable histories. In Figure \ref{fig:sys-arch-r}, yellow arrows depict interactions during the planning phase while green arrows depict interactions during the execution phase.

%

\subsection{Replicated Planning Algorithm}

\begin{algorithm}
	\caption{Planning phase
		\label{alg:planning}}
	\begin{algorithmic}[1]
		\Require{$CTQ$ client transaction queue, $C:$ nodes in the cluster instance, $TC:$ transaction contexts data structures, $s:$ status of the current node }
		\Function{planBatch}{$s,bid,p$}
		\Let{$B$}{$\{\}$}
		\If{isLeader($s$)}
			\While{not $B$.ready()}
				\Let{$m$}{$CTQ$.pop()} 
				\Let{$EQ$}{\textrm{planMessge}($m$, $TC$)}
				\Let{$B$}{$B \oplus EQ$}
			\EndWhile
	
		\State \Return{\textproc{deliverBatch}($B$)} 
			
		\Else
		\State \Return{\textproc{receiveData}($(bid, p)$, $(B, TC)$, \textproc{deliverBatch}($B$))}
		\EndIf
		
		\EndFunction

		\Function{deliverBatch}{$B$}
			\Let{$LEQ$}{$\{q \in B | ~\textrm{isLocal}(q)\}$}
			\Let{$REQ$}{$\{q \in B | ~\textrm{isRemote}(q)\}$}
			\State \textrm{setLocalEQs}($LEQ$)
			\State \textrm{sendRemoteEQs}($REQ$, $C$)
			\State \textproc{replicateData}($(bid, p)$,$(B, TC)$)
		\EndFunction
	\end{algorithmic}
\end{algorithm}

We start by describing the planning algorithm. Algorithm \ref{alg:planning} presents pseudocode for the planning phase. To simplify our presentation, we assume the availability of some global variables. 

$CTQ$ is a variable for the queue holding client transaction. Communication threads push into this queue as they receive transaction messages from clients. 

$C$ is the set of nodes composing the cluster instance. For example, suppose a worker thread running in the planning phase calls \textproc{PlanBatch} on server $L_{00}$. In this case, $C=\{L_{00}, L_{10}, L_{20}\}$.

$TC$ is the data structure that holds the transaction contexts for each planned transaction and holds necessary transaction metadata (e.g., the number of operations and their dependencies). 

The variable $s$ is the status of the node running the worker thread. For instance, for server $\{L_{00}\}, s=L$, and for $\{F_{01}\}, s=F$. When $s=L$, isLeader($s$) $=$ true.

The \textproc{PlanBatch} function is called by each worker thread in each server. Each thread starts by computing the batch identifier $bid$, which determines the order between batches created at different epochs. In our prototype implementation, we use monotonically increasing numbers for batch identifiers. Therefore, for a batch created at epoch $0$, its batch identifier is also $0$. $p$ is the global priority value of the worker thread. As mentioned previously, these values can be either static or dynamic. The only requirement is that no two threads have the same priority. In our implementation, we use static priority values for nodes and worker threads.

Depending on the status of the server node, the planning phase takes two different routes. In case of a node being in the leader set, the worker thread follows Lines $4-9$, reads a message from $CTQ$, and {\em plans} the message (Line $6$). 

When planning a message, the read/write set of the transaction is analyzed. When the read/write set of a transaction includes access to a remote partition, a transaction fragment is created and is placed into a remote EQ destined to the node hosting the target partition. Thus, knowing the full read/write set and the record-to-partition mapping is necessary for planning. The planMessage function returns a set of EQs, and they are {\em merged}  with previously planned EQs. Merging EQs means that transaction fragments accessing the same partition are inserted into the same EQ. 

The call of $B$.ready() at Line $4$ determines when the batch is {\em ready} for delivery. Once the batch is ready, it is delivered to the respective nodes (Lines $14-20$). Local EQs are set in the local partition of the Batch metadata distributed data structure. Remote EQs are sent to their respective nodes to be installed into the remote partitions of the batch metadata. Finally, (in Line $19$) the planned EQs and the transaction contexts are replicated using the replication API (i.e., calls \textproc{replicateData}) to the replica groups (e.g., for $L_{00}$ they are replicated to nodes $F_{01}$ and $F_{02}$).

When the \textproc{planBatch} function call is made by a follower node, it calls the \textproc{receiveData} (Line $11$) and provides the \textproc{deliverBatch} function as the \verb|callback| function. This way when the replicated plans are received, the \textproc{deliverBatch} is called to deliver the planned EQs to node in the replica cluster instance. Using Figure \ref{fig:sys-arch-r} as our example, node $F_{01}$ delivers the replicated batch to nodes $P_1(F1)$ and $P_2(F1)$. 


\subsection{Speculative Queue-oriented Replication Protocol}

As we described in Section \ref{sec:case-studies-qs}, there are many possible replication schemes that can be used with \quecc{}. We propose using the second replication point (from Figure \ref{fig:general-fw}), which is before the execution stage, to perform the replication of EQs and TC using the replication API. The EQs and TC are serialized into a byte string payload and replicated via the replication layer. 

\begin{figure}[!t]
	\centering
	\includegraphics[width=\plotwidth{0.9}{0.6}\linewidth, trim=0cm 0cm 0cm 0cm,clip]{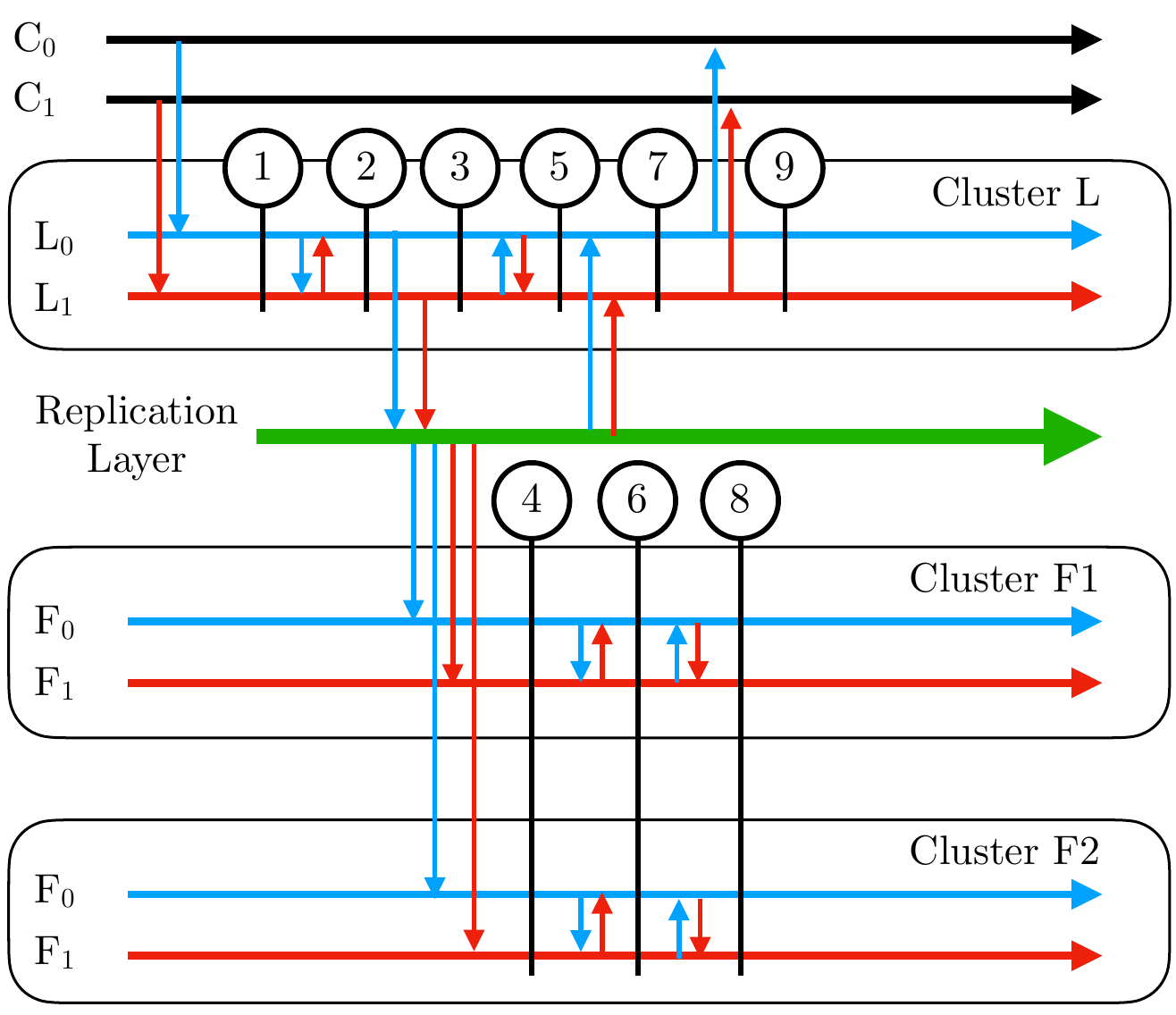}
	\caption{Speculative Execution and Replication Timeline Example}
	\label{fig:repl-comm}
	\vspace{-4mm}
\end{figure}

\textbf{Speculative EQ Replication}.
We propose a queue-oriented speculative approach to replication. The replication is speculative because \quecc{} speculates that the replication would be successful and proceed with the execution phase. The speculation is verified before the commitment stage. Thus, we effectively hide the EQs' replication latency by performing it concurrently with EQs' execution. 

Figure \ref{fig:repl-comm} illustrates an example timeline (time goes from left to right) of the replicated transaction processing workflow. A client is associated with a server in the leader set that is considered the {\em home server} for that client. At the start (before \circled{1}), clients send transactions to home servers. In this example, $C_0$ sends transactions to server $L_{00}$ (i.e., the blue partition) while $C_1$ sends transactions to $L_{10}$ (i.e., the red partition). At \circled{1}, the batch is ready for delivery, and both leaders send their planned remote EQs to the other node. At \circled{2}, the leaders submit their replication requests to the replication layer (i.e., using \textproc{replicateData} API call). Leader nodes start the execution phase immediately without waiting for the outcome of the replication at \circled{3} following the speculative replication approach. Between \circled{3} and \circled{5} processing acknowledgments messages are exchanged within the leader cluster instance. The replication layer ensures that replication requests are delivered reliably to the followers by \circled{4} (i.e., using \textproc{receiveData} API call). Replicated plans are exchanged in the follower clusters by \circled{6}. The replication layer responds to the leader set nodes by \circled{7}. After \circled{7}, leader nodes safely proceed with the commitment stage and commit transactions. At the follower clusters, the execution phase starts at \circled{6}, and the commitment stage starts at \circled{8}. The commitment stage at the follower nodes requires acknowledgments from nodes in their cluster instance to ensure that multi-partitioned transactions are processed successfully by all participating partitions (i.e., between \circled{6} and \circled{8}). By \circled{9} leader nodes respond to clients. 

\textbf{Discussion}
Note that the replication layer can respond to the replication request by the leader set of nodes before the followers receive the EQs. This invariant is stated as follows:

\begin{invariant}[Replication Invariant]
	\label{inv:repl}
The leader nodes receive acknowledgments of their replication request from the replication layer if and only if the replication layer guarantees that followers eventually receive replicated data.
\end{invariant}

It is the responsibility of the replication layer implementation to ensure the eventual delivery of replicated data. The above invariant allows some flexibility in implementing the replication layer, which can be a middleware-based or an integrated implementation.

\begin{figure}
	\centering
	\includegraphics[width=\plotwidth{0.98}{0.7}\linewidth]{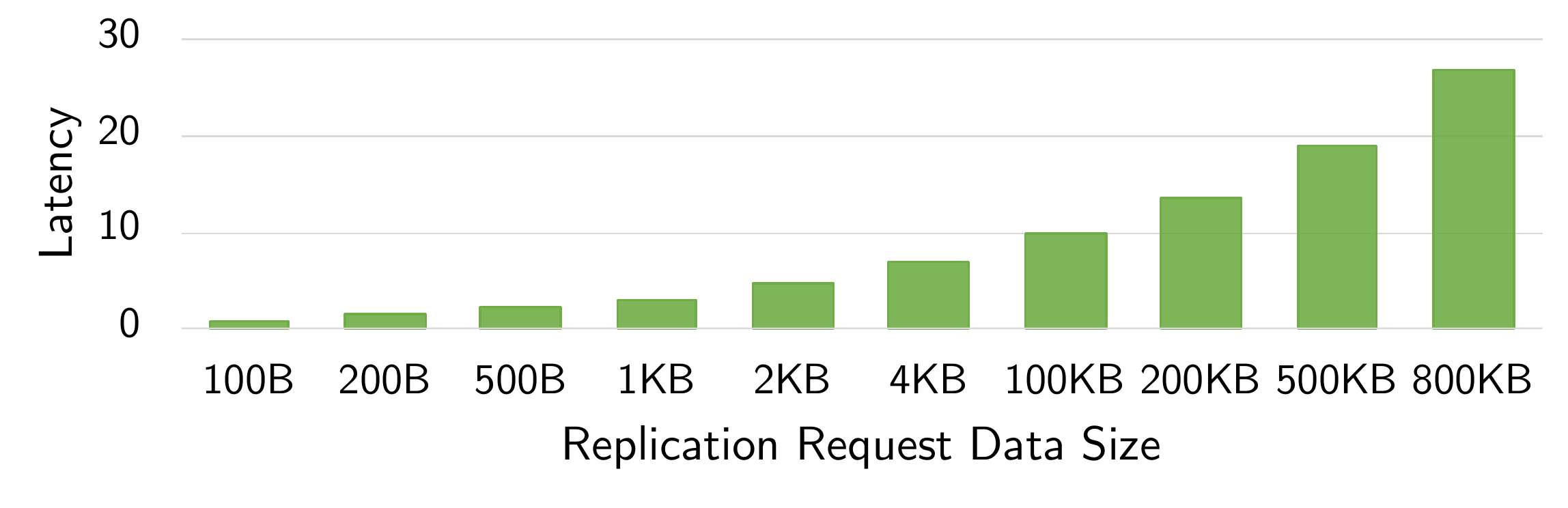}
	\vspace{-4mm}
	\caption{Zookeeper latency micro-benchmark. Latency is measured in milliseconds.}
	\label{fig:zklatency}
	\vspace{-5mm}
\end{figure}


\textbf{Replication Payload Compression}. 
Replication payload is a function of the batch size. Therefore for large batch sizes, they can be in the order of a few hundred kilobytes. At this scale of message sizes, the latency can be undesirably too high. Figure \ref{fig:zklatency} shows the result of a micro-benchmark of submitting $100$ concurrent requests to Zookeeper (as a replication service) while varying the payload size from $100$ bytes to $800$ kilobytes. Note that Zookeeper can only support a maximum of $1$ megabytes of data stored as a single Zookeeper node. We can observe that for large message sizes, the latency can reach up to $26$ milliseconds. 

To reduce the payload size of replicated data in \quecc{}, we compress replication data only using Snappy \cite{snappy_2021}, we observed a reduction of payload sizes by $60\%$.

\begin{algorithm}
	\caption{Execution algorithm
		\label{alg:execution}}
	\begin{algorithmic}[1]
		\Require{$BM$ batch metadata}
		\Function{executeBatch}{$tid$}
		\While{not $BM$.done()}
		\Let{$q$}{$BM$.\textproc{getTop($tid$)}} \label{lst:line:gettop}
		\While{not $q$.empty()}
		\Let{$f$}{$q$.pop()}			
		\State \textproc{executeTF($f$)} \label{lst:line:extf}
		\State \textproc{resolveDependency($f$)} \label{lst:line:rtfdep}
		\EndWhile
		\If{isRemote($q$)}
		\State \textproc{sendAck($q$)}\label{lst:line:sendack}
		\Else
		\State \textproc{updateTC($q$)}\label{lst:line:updateTC}
		\EndIf
		\EndWhile		
		\EndFunction		
	\end{algorithmic}
\end{algorithm}

\subsection{Speculative Execution Algorithm}
In this section, we present the execution algorithm in \quecc{}. The algorithm is simple, and its pseudo-code is presented in Algorithm \ref{alg:execution}. $BM$ is a reference to the Batch metadata data structure, which is assumed to have global access. A worker thread at the execution phase keeps working on executing EQs until all EQs are processed. It retrieves the next available EQ at Line \ref{lst:line:gettop} using \textproc{getTop}. The returned EQ must satisfy the following conditions:

\begin{enumerate}
\item Condition 1: if a worker thread $i$ gets EQ $q$ using \textproc{getTop}, no other worker thread $j$ gets $q$. The thread identifier $tid$ is used to ensure this condition is satisfied.
\item Condition 2: The read/write sets of transaction fragments in $q$ \textbf{do not} overlap with read/write sets of any other transaction fragment in EQs that remains in $BM$. 
\item Condition 3: $q$ has the highest priority in $BM$
\end{enumerate}

These conditions ensure the following global execution priority invariant is maintained across all worker threads in the cluster when executing Line \ref{lst:line:extf}. 

\begin{invariant}[Global Execution Priority Invariant]
	\label{inv:exec}
Across all nodes in a cluster, transaction fragments from higher priority EQ \textbf{are always} executed before transaction fragments from lower priority EQs.	 
\end{invariant}

After executing a transaction fragment, we need to resolve any data dependencies of that fragment (Line \ref{lst:line:rtfdep}). An example of a data dependency is a transaction fragment that performs the following operation $f: x = x + y$. In this case, $x$ and $y$ are records that belong to different partitions. Reading record $y$ is needed to resolve the dependency of computing the new value of $x$. Hence, to resolve the dependency on $x$, we need to send the value of $y$ to the node executing the transaction fragment $f$, which is the node that holds record $x$. The transaction contexts maintain the state of transactions, and they are updated during execution. If $q$ is a remote EQ, an ACK is sent to the original planning node (Line \ref{lst:line:sendack}). When the ACK is received by the planning node, the transaction contexts of relevant transactions are updated (e.g., updating the number of fragments that are completed). If $q$ is local (i.e., $q$ is planned by the same node that executed $q$), the transaction contexts are updated locally (Line \ref{lst:line:updateTC}). 

This execution is speculative because transaction fragments from different transactions are executed and the commitment is done at a later stage. For example, an EQ $q$ can contain transaction fragments from  $f_1$ and $f_2$ from transactions $t_1$ and $t_2$, respectively. A worker thread executes $f_1$ followed by $f_2$. However, $t_1$ is committed later after executing $f_2$. 

\textbf{Discussion} 
The main problem associated with speculative execution is the notion of cascading aborts \cite{qadah_quecc:_2018}. DTP systems abort only if the transaction has explicit abort logic. For example, a transaction $t_i$ that makes a product purchase would {\em abort if the product's $stock == 0$}. Any transaction that conflicts with $t_i$ will also abort if it reads any records updated by $t_i$ because values written by $t_i$ are not committed and should not be visible. \quecc{} keeps track of transaction conflict information in the form of a dependency graph. The graph is made available to the commitment stage so that transactions are committed according to the correct isolation level. We assume serializable isolation in throughout this paper. However, as shown in \cite{qadah_q-store_2020}, we can also support other isolation levels.

\subsection{Commitment Algorithm}

For transaction commitment, the original planner node act as the transaction coordinator for all transaction it planned. Thus, it requires receiving ACKs for all remote EQs. These ACK messages are communicated to the transaction coordinator node as remote EQs are executed. As an illustration, in Figure \ref{fig:repl-comm}, they are communicated between \circled{3} and \circled{5} for the leader set, and between \circled{6}, and \circled{8} for the replica sets. 

\begin{algorithm}
	\caption{Commitment Algorithm
		\label{alg:commitment}}
	\begin{algorithmic}[1]
		\Require{$BM$ batch metadata}
		\Function{commitBatch}{$tid$} \label{lst:line:ca:func}
		\Let{$T$}{$BM.TC$.\textproc{getTransactions($tid$)}} \label{lst:line:ca:gettxns}
		\Let{$P$}{empty FIFO queue for pending transactions} \label{lst:line:ca:pendq_const}
		\For{$t \in T$} \label{lst:line:ca:commit_iter_start}
			\Let{$status$}{\textproc{commitTxn($t$)}} \label{lst:line:ca:commitTxn1}
			\If{not $status$} \label{lst:line:ca:nocommit}
				\State $P$.push($t$) \label{lst:line:ca:pendq_push}
			\EndIf
		\EndFor \label{lst:line:ca:commit_iter_end}
		
		\While{not $P$.empty()} \label{lst:line:ca:start_pending_check}
			\Let{$status$}{\textproc{commitTxn($P$.head)}} \label{lst:line:ca:commitTxn2} 
			\If{$status$}
				\State $P$.pop() \label{lst:line:ca:pendq_pop}
			\EndIf
		\EndWhile \label{lst:line:ca:end_pending_check}
		\EndFunction		
	\end{algorithmic}
\end{algorithm}

Algorithm \ref{alg:commitment} shows the pseudo-code for the commitment algorithm used by the transaction coordinator nodes. Because a leader node can run multiple planning producing different sets of plans, each planning thread is identified by the $tid$. Thus, the $tid$ is used to commit a transactions planned by a specific planner (Line \ref{lst:line:ca:gettxns}). In Line \ref{lst:line:ca:pendq_const}, $P$ is initialized to an empty FIFO queue to hold transactions pending commit. The order of commitment is concerned only with conflicting transactions. Non-conflicting transactions can commit in any order, and our algorithm allows. In Lines \ref{lst:line:ca:commit_iter_start} to \ref{lst:line:ca:commit_iter_end}, we perform a single iteration to commit transactions. In Line \ref{lst:line:ca:commitTxn1}, the \textproc{commitTxn} function checks if the transaction can commit (i.e., all of its fragments are executed successfully). It returns {\em true} if the transaction $t$ is committed, and {\em false} otherwise. If a transaction cannot commit at this time (Line \ref{lst:line:ca:nocommit}), it is pushed into the pending transaction queue $P$ for a later check, which happens in Lines \ref{lst:line:ca:start_pending_check} to \ref{lst:line:ca:end_pending_check}. A transaction is checked at the head of the queue without removing it (Line \ref{lst:line:ca:commitTxn2}). It is only removed if it is committed (Line \ref{lst:line:ca:pendq_pop}). 



Note that a single-threaded implementation of the commitment algorithm can join all transactions into a single set for commitment. The only requirement is that to ensure that the correct commit order of conflicting transactions is preserved. 

Regardless of the implementation of the commitment algorithm, the commit stage needs to adhere to the following invariant. We use $po(t)$ to denote planning order, and $co(t)$ to denote the commitment order of transaction $t$, respectively.  

\begin{invariant}[Commitment Invariant]
	\label{inv:commit}
For any two conflicting transactions $t_i$ and $t_j$, $co(t_j) > co(t_i) \iff po(t_j) > po(t_i)$.
\end{invariant}

\subsection{Correctness}

Based on the above three invariants, we state the following theorem and provide a proof sketch.

\begin{theorem}
QR-Store's transaction processing protocol is safe and strictly serializable.
\end{theorem}

\begin{proof}
It follows from the three invariants stated above that the \quecc{} processes transactions with strict serializability.  Invariant \ref{inv:commit} ensures that the commitment order of conflicting transactions follows the planning order. Planning threads impose ordering by using the ordering property of queues in EQs. The order between EQs planned by different planning threads is determined by the priority order of the planning threads. Thus, there is a global partial order of all transaction fragments, which is preserved by Invariants \ref{inv:exec} and \ref{inv:commit}. Furthermore, because the commitment stage does not start until the replication requests are acknowledged according to Invariant \ref{inv:repl}, the transaction commitment is safe.  
\end{proof}


\subsection{Latency Model}\label{sec:latency-model}

\begin{figure}
	\centering
	\includegraphics[width=\plotwidth{0.98}{0.7}\linewidth,trim=0cm 0cm 0cm 0cm,clip]{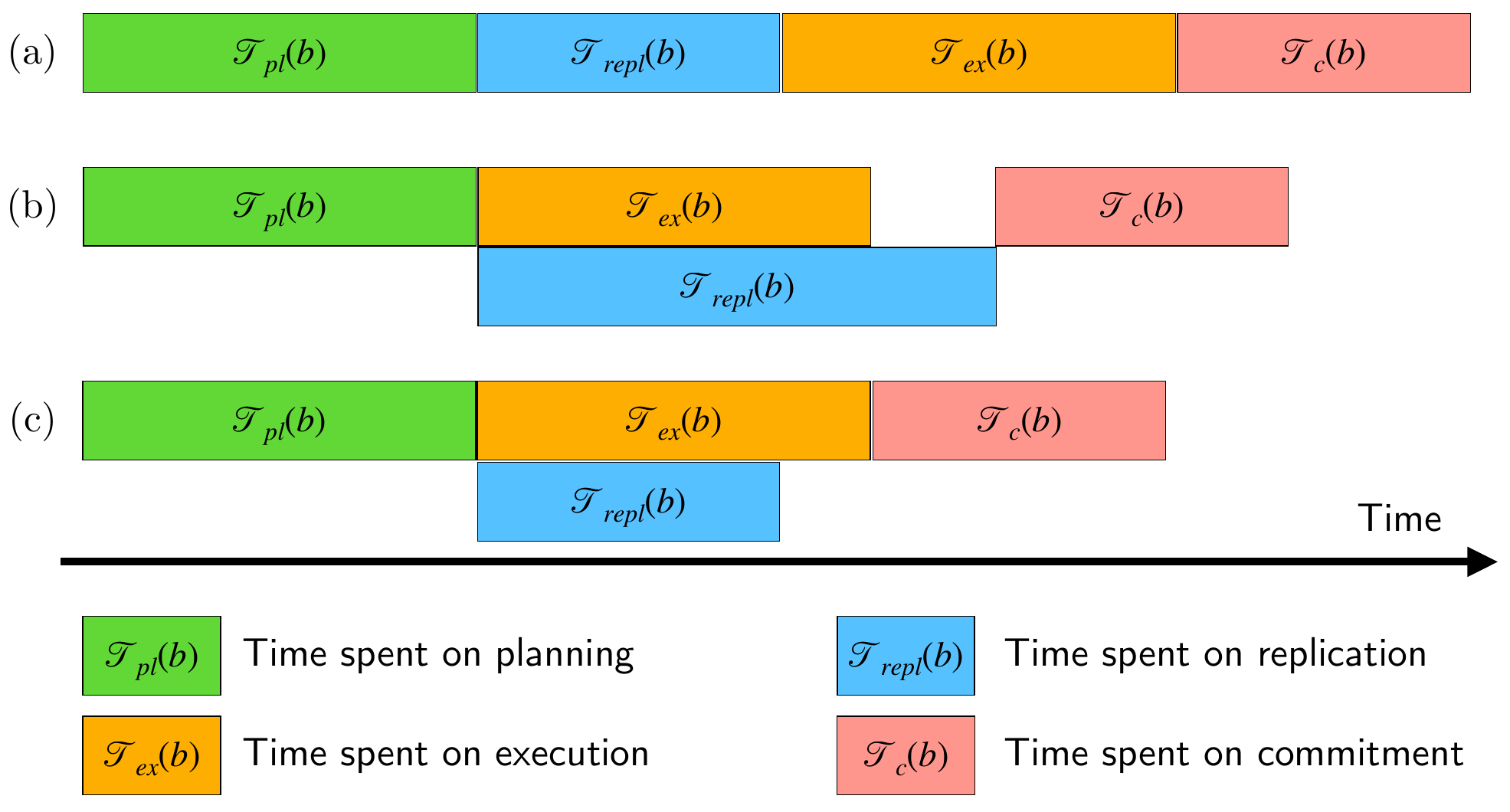}
	\caption[]{Illustrating Replication Overhead in QR-Store}
	\label{fig:block-timeline}
	\vspace{-4mm}
\end{figure}

In this section, we model the latency for our queue-oriented transaction processing with replication. The key idea of performing speculative replication is to hide the replication latency overhead. The following equation models the latency of completing the processing of a single batch which is denoted as $\mathcal{T}_b$.

\begin{equation}
	\mathscr{T}(b) = \mathscr{T}_{pl}(b) +  max(\mathscr{T}_{deliv.} + \mathscr{T}_{ex}(b), \mathscr{T}_{repl}(b)) + \mathscr{T}_{c}(b)
\end{equation}

\begin{itemize}
	\item $\mathscr{T}_{pl}(b)$ is the time spent in the planning phase for batch $b$
	\item $\mathscr{T}_{deliv.}$ is the time spent on delivering remote EQ messages for batch $b$
	\item $\mathscr{T}_{ex}$ is the time spent in the execution stage for batch $b$
	\item $\mathscr{T}_{repl}$ is the time spent waiting for the replication to be confirmed for batch $b$ by the replication layer
	\item $\mathscr{T}_{c}$ is the time spent on committing transactions for batch $b$
\end{itemize}

In Figure \ref{fig:block-timeline}, we show a visualization of three cases of queue-oriented replication. (a) depicts using synchronous replication in \quecc{}. In this case, the replication request must be acknowledged before we start the execution phase. Thus, the overhead of replication is directly added to the latency of processing a batch in \quecc{}. (b) and (c) in Figure \ref{fig:block-timeline} are the two cases of using the speculative replication approach. In (b), the replication takes longer than the execution, which forces the execution threads to wait for the replication confirmation before starting the commit stage. The optimal case is depicted by (c), which totally hides the replication latency, while in (b), the replication overhead is partially hidden.



\subsection{Logging and Recovery}

All proposed DTP systems assume a deterministic stored procedure transaction model (e.g., 
\cite{kallman_h-store_2008,thomson_calvin_2012, qadah_quecc:_2018, qadah_queue-oriented_2019, qadah_q-store_2020}). 
Furthermore, the stored procedure model assumes that the transaction logic is deterministic. In other words, the output is the same as long as the procedure is given the same input. 

DTP systems use a combination of checkpointing and command-logging to facilitate logging and recovery. With command-logging, only the input of the transactions is logged, and on recovery, the log is applied from the latest stable checkpoint. Checkpointing can be done asynchronously to the disk to avoid blocking the transaction processing workflow.  

In \quecc{}, when a leader node crashes and resumes operation, the first step is to determine the new leader. The next step is to request all EQs from the current leader since the last stable checkpoint and execute them to recover the database partition state. After that, it acts as a follower by getting its plans from the replication layer. 

Follower nodes detect leader nodes crashes via heartbeat messages exchanged periodically between leaders and followers. When followers detect that a leader has failed, they run a leader election process among them. Once a new leader is elected, the newly elected leader node requests leadership ACKs from other followers to determine the last committed batch. It requests any missing queues and replays them. 

Using command-logging only with \quecc{} introduces a recovery challenge. First, logged commands need to be planned again, which increases the latency of recovery. Second, when recovering multi-partition transactions, participation from all partitions is required to resolve data dependencies.

Therefore, instead of using command-logging and simply log transaction inputs, \quecc{} create special write-only EQs that contain the last write operation of records accessed by planned EQs. These write-only EQs are logged to facilitate recovery. Thus, to recover a node's state, it is sufficient to request these write-only EQs from other nodes. This approach also resolves data dependencies associated with replaying multi-partition transactions because the logged value does not have any data dependencies.

\section{Implementation}
\label{sec:qrep-impl}
This section discusses the implementation aspects of the replication layer and optimizations related to synchronization granularity.  We show the impact of these implementations and optimizations in Section \ref{sec:qrep-eval}.

\subsection{Replication Layer Abstraction}
As mentioned in Section \ref{sec:genfw} that we provide a simple API abstraction for the replication layer. We now give some details on the API, which consists of the following two functions.

\textproc{replicateData}\verb|(meta_data, data, [callback])|.   
This is an asynchronous function that is called by the leader set of nodes (the dotted arrow in Figure \ref{fig:general-fw} originating from the red circles). \verb|meta_data| parameter can have some identifying information of the data being replicated (e.g., a batch identifier). The \verb|data| parameter is a byte string of the data being replicated. \verb|callback| parameter is a function that is called after the \textproc{replicateData} function completes. Since the function is asynchronous, the callback function provides a way to perform some actions (e.g., error handling).

\textproc{receiveData}\verb|(meta_data, callback)|.  
This is also an asynchronous function, and it is called by replicas to receive the replicated data. (the solid arrow in Figure \ref{fig:general-fw} originating from the green replication layer) The \verb|meta_data| parameter can have some identifying information about the replicated data from the replica's perspective. For example, it can include the expected batch identifier. The \verb|callback| parameter is a function that is called with and passed the replicated data. It is used to construct the DTP system's representation of the replicated data from byte string passed to \textproc{replicateData}.

%

\subsection{Replication Layer Implementations}
Our prototype provides two implementations of the replication layer, and we describe these implementations in this section. 

\textbf{Middleware Replication} The first one uses Zookeeper \cite{apache_zookeeper_website} as a middleware to implement the replication Layer. Leader servers and replica servers act as clients to the Zookeeper cluster. While leader servers make write requests, replica servers make read requests to get the replicated data. The Zookeeper cluster is a highly available system, and it does not constitute a single point of failure because it uses its internal replication and consensus protocols to ensure correct fail-over. The consensus protocol used by Zookeeper is called \zab{}\cite{junqueira_zab_2011} The advantage of this approach is that it simplifies the replication layer shim implementations at the server nodes. The disadvantages include adding an overhead of the middleware to the transaction processing protocol. This approach is adopted by \calvin{} in their original proposal \cite{thomson_calvin_2012} and also in our experiments.  

\textbf{Integrated Replication}
The second implementation of the replication layer integrates a quorum-based replication protocol with the transaction processing protocol and effectively eliminates the middleware overhead. Our implementation is based on \raft{}\cite{ongaro_search_2014} and Viewstamped replication \cite{oki_viewstamped_1988}. With this implementation, the leader server nodes send replicated data messages to replicas. On receiving replicated data messages, replicas reply with acknowledgment messages to leader servers. Depending on the number of node failures to tolerate, there is a minimum number of acknowledgment messages that confirm a successful replication. Let the number of node failures be denoted as $f$ such that the total number replica for a given node is $n = 2f+1$.  The number of acknowledgment messages is $f+1$.


\subsection{Synchronization Granularity}
In \quecc{}, we support multiple granularities of transaction processing synchronization within a cluster. A coarse-grained synchronization puts node-level barriers between batches. Hence, before any thread can start processing the next batch, it has to wait for all other nodes to finish processing the current batch. At the cost of additional implementation complexity, it is possible to have a more fine-grained synchronization at the thread level. This thread-level synchronization allows worker threads to start the planning phase of the next batch as soon as other nodes acknowledge the execution of their respective planned EQs. Thread-level synchronization improves the concurrency of the phases across batches. Our previous prototype implementation for \qstore{} \cite{qadah_q-store_2020} uses node-level synchronization while our current prototype uses thread-level.


%


\section{Evaluation}
\label{sec:qrep-eval}

In this section, we present our experimental evaluation. We implement the three case studies that we discussed earlier in Section \ref{sec:genfw}. We use the first replication point for \calvin{}, and use the second replication point for \quecc{} and a fully replicated version of \qcmw{} (denoted as \qsonenode{}). The experimental study of the other replication points admitted by our proposed framework in Section \ref{sec:genfw} remains future work. 

We mainly focus on \quecc{} with various replication factors configured. The current prototype of \quecc{} is the optimized and improved version of our previous work presented in \cite{qadah_q-store_2020}. 
In our comparison with \calvin{}, we use Zookeeper as the implementation of the replication layer as it is originally presented in \cite{thomson_calvin_2012}.

\subsection{Experimental Setup}
We use up to $64$ \verb|c2-standard-8| instances on Google Cloud Platform to run our experiments. These instances have $8$ vCPUs, $32$GiB of RAM, and the default egress network bandwidth available is $16$Gbps. Each node runs Ubuntu 18.04 (bionic beaver), and the codebase is compiled with the $-O3$ compiler optimization flag. 
We use four worker threads and four communication threads. Threads are pinned to cores to minimize potential variance due to the operating system. 

Furthermore, our \calvin{}'s implementation dedicates one worker thread for the sequencer role and another worker thread for the scheduler role. This configuration leaves two threads for processing transactions \footnote{using additional worker threads reduces Calvin's performance by nearly $50\%$ according to our observations}.

Each data point is the average of three trials. Each trial consists of a warm-up phase of $60$ seconds where measurements are not collected, followed by a measurement phase of $60$ seconds.

\begin{table}[t]
	\centering 
	\caption{System and workload configuration parameters. Default values are in parenthesis. Default values are used unless stated otherwise.}  
	\label{tab:ycsb-wl-conf-qrep}
	\resizebox{\plotwidth{0.98}{0.7}\columnwidth}{!}{
		\begin{tabular}{ l l l }
			\hline
			P\# & \textbf{Parameter Name} & \textbf{Possible Parameter Values} \\ 
			$P1$ & $\%$ of multi-partition txns. & $0\%, 10\%, 15\%, (50\%), 75\%, 100\% $ 		
			\\ 
			$P2$ &Zipfian's theta & $(0.0),0.4,0.6,0.8,0.9,0.99$ \\ 
			
			
			$P3$ & Operations/txn. & $2,4,8,12,(16)$ \\ 
			
			
			$P4$ & Batch sizes
			& $2K, 5K, 10K, (20K), 40K, 80K, 100K$ \\
			
			$P5$ & Server nodes counts & $2,4,8,(16)$ \\ 
			
			$P6$ & Replication factor & $(0),(1),(2),4,6,8$ \\ 
			\hline
			
		\end{tabular}
	}
\end{table}

\textbf{System and workload parameters}
In Table \ref{tab:ycsb-wl-conf-qrep}, we present all system and workload parameters used in our experiments. $P1$ is the percentage of multi-partition transactions (MPTs) in the workload. An MPT accesses more than one partition and requires a distributed commit protocol. $P2$ is the parameter that controls the access distribution of transactions. Higher values of $\theta$ make the access skewed to a small set of records. 
$P3$ is the parameter that controls the number of operations per transaction and requires transaction atomicity when there is more than one operation in a transaction. A transaction with a single operation is atomic by definition. 
$P4$ is the size of transaction batches that are processed by queue-oriented transaction processing systems such as \quecc{}, \qstore{} and \qcmw{}. $P5$ controls the number of server nodes used in a cluster. A cluster of nodes forms an instance of the database, and each node manages a single partition. By default, we use one client node per server node. $P6$ represents the number of replicas used per cluster of servers. For example, a value of $2$ means that there are $2$ additional replicated database instances.

\subsection{Experimental Results}

We first study the impact of replication. We use three configurations. As a baseline, we use \qstore{} which does not perform replication. \qsone{} and \qstwo{} has a replication factor of $1$ and $2$, respectively. 

\begin{figure}[!t]
	\includegraphics[width=\plotwidth{0.9}{0.6}\columnwidth]{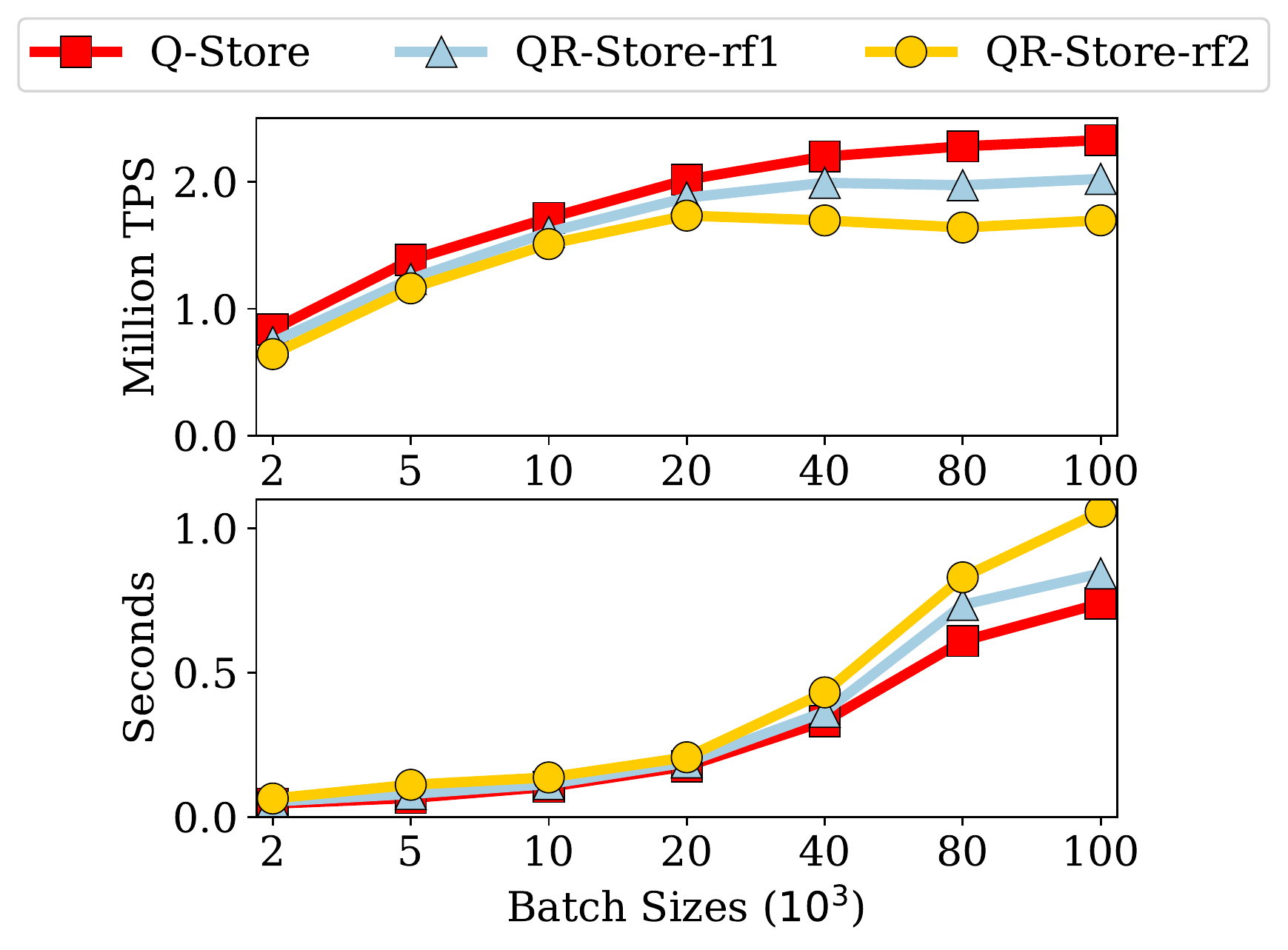}
	\centering
	\vspace{-3mm}
	\caption{Varying Batch Size}
	\label{fig:eval-bsizevar-qrep}
	\vspace{-4mm}
\end{figure} 

\textbf{Varying the batch size} 
In this set of experiments, we want to understand the effect of the batch size on the performance of \quecc{}. We use $50\%$ MPT transactions in the workload, uniform access distribution (i.e., a value of $\theta=0.0$), $50\%$ update operation per transaction, $16$ operation per transaction, and each transaction access $8$ partitions. The number of server nodes is $16$. 

Figure \ref{fig:eval-bsizevar-qrep}, shows the system throughput and the $99^{th}$ percentile latency of transaction processing. \qszero{} which is the configuration without replication performs the best both in terms of throughput and latency (up to $2.3$ million TPS executed in under $750$ milliseconds). While \qszero{}'s throughput performance keeps increasing with batch sizes greater than $20K$ transactions, the latency also increases significantly. With larger batches, worker threads spend more time planning and executing transactions. Also, the size of messages exchanged between the server nodes within the cluster becomes larger. Moreover, beyond the $20K$ batch size, the gap in performance \qszero{}, and the replicated configurations (i.e., \qsone{} and \qstwo{}) becomes more significant because the leader cluster needs to prepare and replicate larger plans. With \qstwo{} the number of messages that are sent by the leader nodes is twice the number sent by the \qsone{} configuration. Hence, as these messages become larger, the computation and communication requirements increase. For example, for $rf=2$ the $99^{th}$ percentile latency increases from $30\%$ at $40K$ batches to $43\%$ at $100K$. Notably, at $20K$ batches, the latency overhead is only $16\%$ in this workload configuration.



\begin{figure}[!t]
	\includegraphics[width=\plotwidth{0.82}{0.6}\columnwidth]{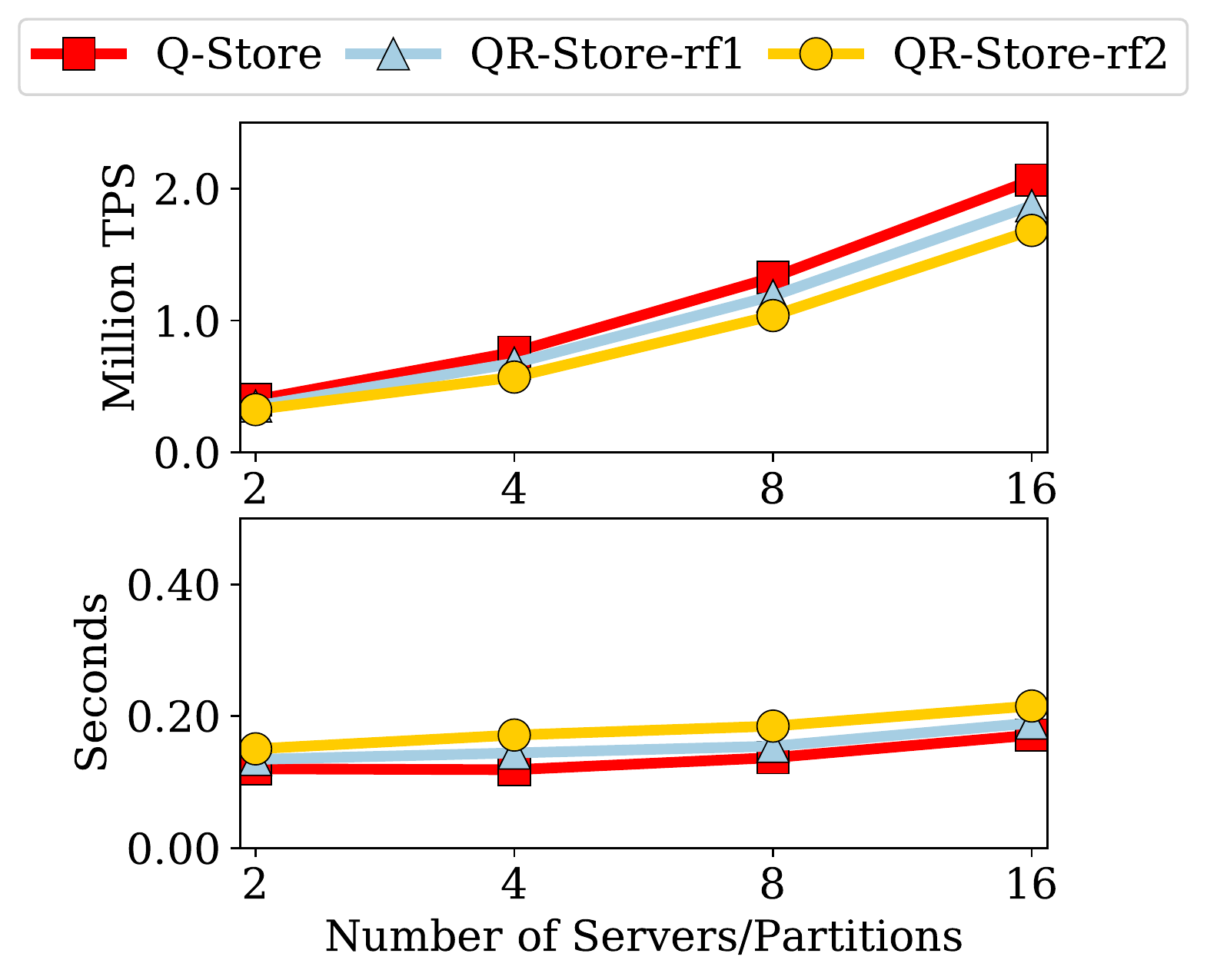}
	\centering
	\vspace{-3mm}
	\caption{Scalability with increasing the number of servers/partitions in a cluster instance.}
	\label{fig:eval-ncvar-qrep}
	\vspace{-4mm}
\end{figure}


\textbf{Scalability when increasing number of server nodes} 
In this set of experiments, the percentage of MPT is $50\%$, the Zipfian theta parameter is set to $\theta = 0.0$, the percentage of write operations is $50\%$, the number of operations per transaction is $16$, the batch size is set to $20K$ transactions, and we force each transaction to access all available partitions. In other words, MPTs will always access all servers. We vary the number of server nodes from 2 to 16.  

Figure \ref{fig:eval-ncvar-qrep} shows that all configurations scale linearly as we add more nodes into the server cluster. The linear scaling is because operations in each transaction are processed in parallel by all available nodes. Notably, the throughput performance reaches $2$, $1.8$, and $1.7$ million transactions per second for \qszero{}, \qsone{}, and \qstwo{}, respectively. The $99^{th}$ percentile latency remains under $216$ milliseconds.  

When using higher replication factors (i.e., \qsone{}, and \qstwo{}), the impact of replication becomes larger. In our experiments with $16$ server nodes, the performance drops by $15\%$. We believe that this is a reasonable cost to ensure fault tolerance. 




\begin{figure}[!t]
	\includegraphics[width=\plotwidth{0.82}{0.6}\columnwidth]{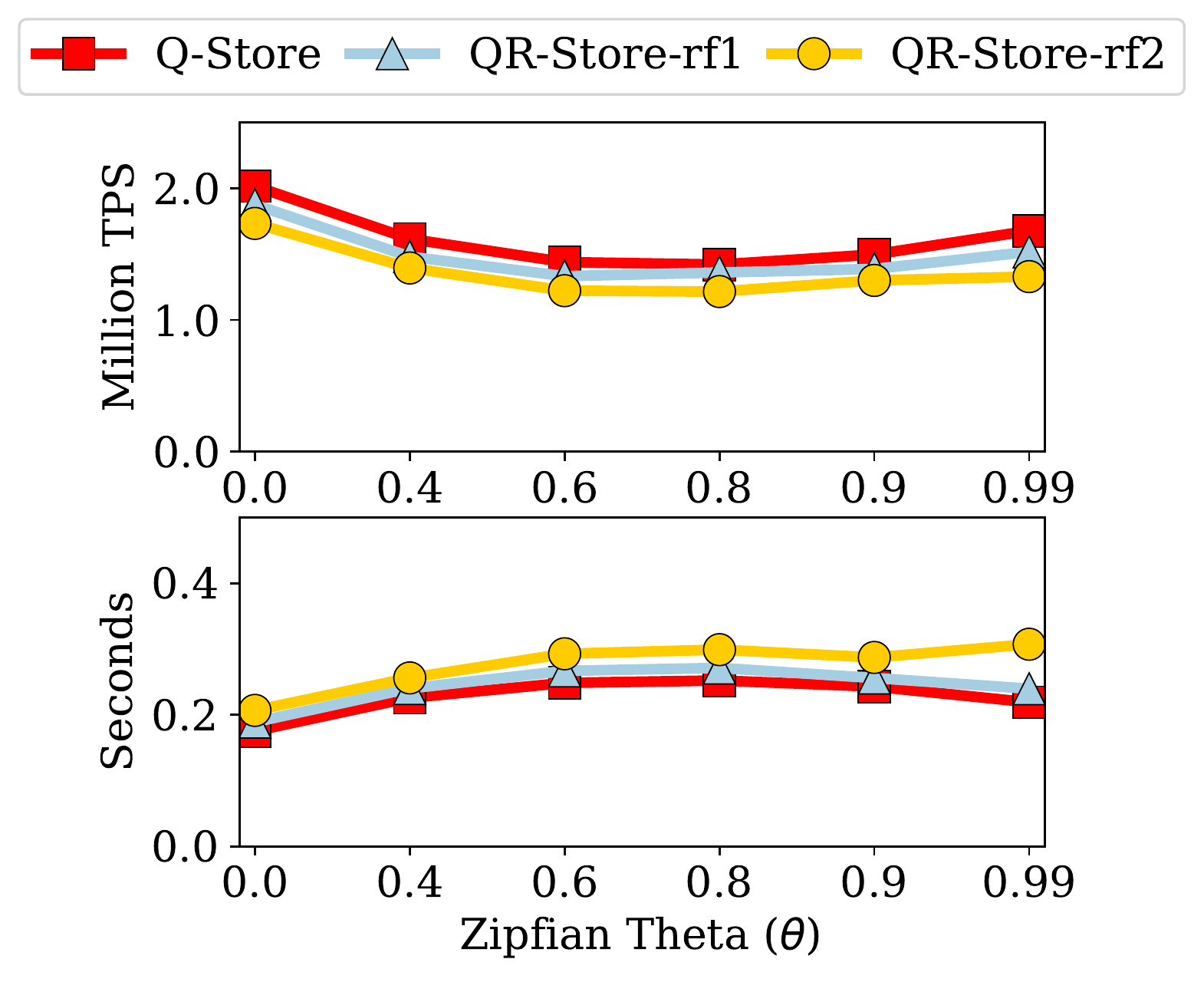}
	\centering
	\vspace{-3mm}
	\caption{Varying Data Access Contention}
	\label{fig:eval-thetavar-qrep}
	\vspace{-4mm}
\end{figure}

\textbf{Varying Data Access Contention} 
In Figure \ref{fig:eval-thetavar-qrep}, we vary the value of the Zipfian distribution $\theta$ from $0.0$ to $0.99$. Using $0.0$ for $\theta$ creates uniform access to database records, while a value $0.99$ creates extremely skewed access resulting in an increased contention on database records. 

Our queue-oriented approach naturalizes the high contention because the operations accessing the same set of records are placed in the same EQ and are executed by the same worker thread. However, as shown in Figure \ref{fig:eval-thetavar-qrep}, we observe a decrease in performance when there is medium contention (i.e., $0.4-0.8$). For \qszero{}, the throughput drops by $19-29\%$, and the latency increases by $28-39\%$. The throughput drops by $19-26\%$ and $17-26\%$, and the latency increases by $23-42\%$ and $26-42\%$ for \qsone{} and \qstwo{}, respectively. At medium contention, some EQs contain more operations than others, which increases the execution time. However, at high contention (i.e., $0.9-0.99$), the performance gets better because the caching becomes more effective as most operations in the large queues access a small set of records. Notably, at low contention (i.e., $\theta = 0.0$), the overhead of replication is $6\%$ and $13\%$ for \qsone{} and \qstwo{}, respectively.

\begin{figure}[!t]
	\includegraphics[width=\plotwidth{0.82}{0.6}\columnwidth]{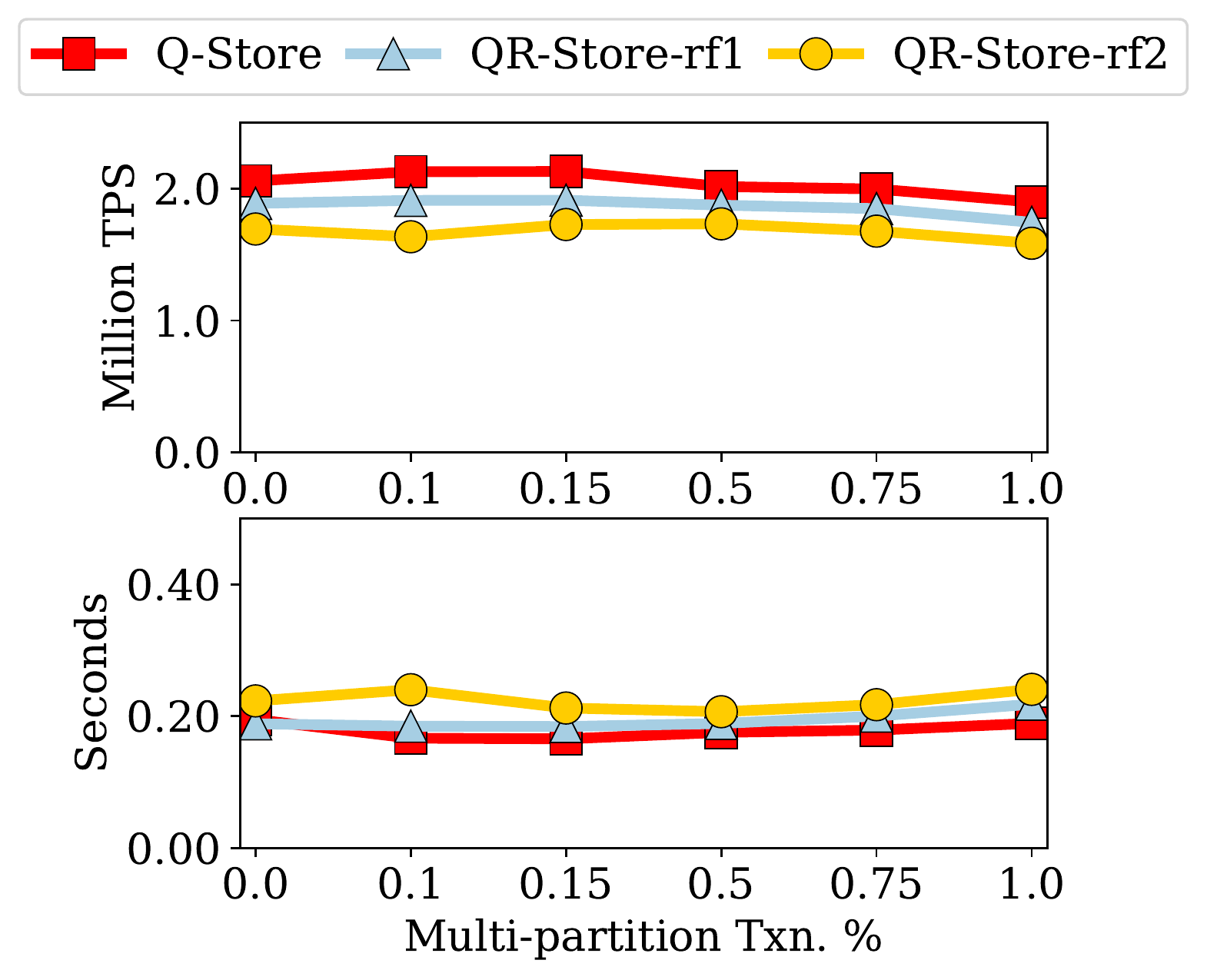}
	\centering
	\vspace{-3mm}
	\caption{Varying Multi-partition transaction rate}
	\label{fig:eval-mptvar-qrep}
\end{figure} 

\textbf{Varying MPT percentage} 
Now, we look into the effect of increasing the percentage of the multi-partition transactions in the workload. For this set of experiments, we use uniform access and enforce each MPT to access $8$ partitions of the $16$ partitions.  The update operation percentage remains at $50\%$. Increasing the MPT percentage increases the sizes of remote EQs that are planned. The throughput performance gets better at a low ratio of MPT in the workload by $6.5\%, 10.5\%$ and $24\%$ for \qszero{}, \qsone{}, and \qstwo{}, respectively, because some operations are executed remotely by other nodes which reduced the load on local worker threads. However, as we increase the ratio, the performance starts dropping to even below the performance of a pure single partition workload because remote executions take longer times to be acknowledged.

\begin{figure}[!t]
	\includegraphics[width=\plotwidth{0.82}{0.6}\columnwidth]{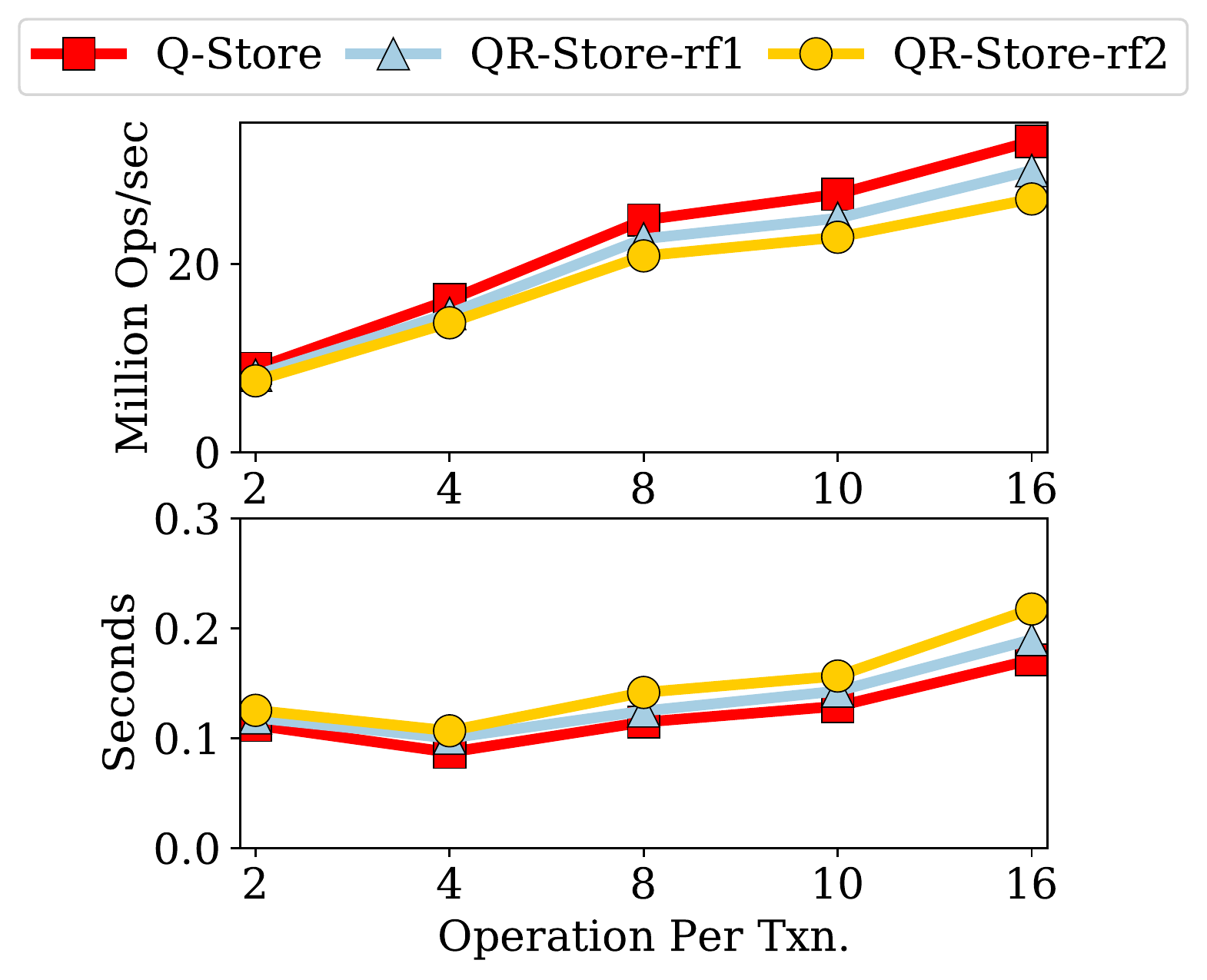}
	\centering
	\vspace{-3mm}
	\caption{Varying operation per transaction}
	\label{fig:eval-optvar-qrep}
	\vspace{-4mm}
\end{figure} 

\textbf{Varying the number of operations per transaction} 
The number of operations per transaction represents the transaction size. Again, we fix the other system and workload parameters to their default values and vary the number of operations per transaction. In Figure \ref{fig:eval-optvar-qrep}, we use the number of operations processed per second instead of the number of operations. All configurations scale their throughput performance linearly as we increase the number of operations, and the throughput performance reaches up to $33, 30$, and $27$ million operations per second for \qszero{}, \qsone{} and \qstwo{}, respectively. The $99^{th}$ percentile latency is between $110-220$ milliseconds. With a low number of operations per transaction, the communication overhead is more significant. As the number of operations per transaction increases, the EQs become larger and more efficient to execute. With replication the throughput performance drops by $9\%$ and $18\%$ for \qsone{} and \qstwo{}, respectively.
 
%
%
%
%

\begin{figure}[!t]
	\includegraphics[width=\plotwidth{0.82}{0.6}\columnwidth]{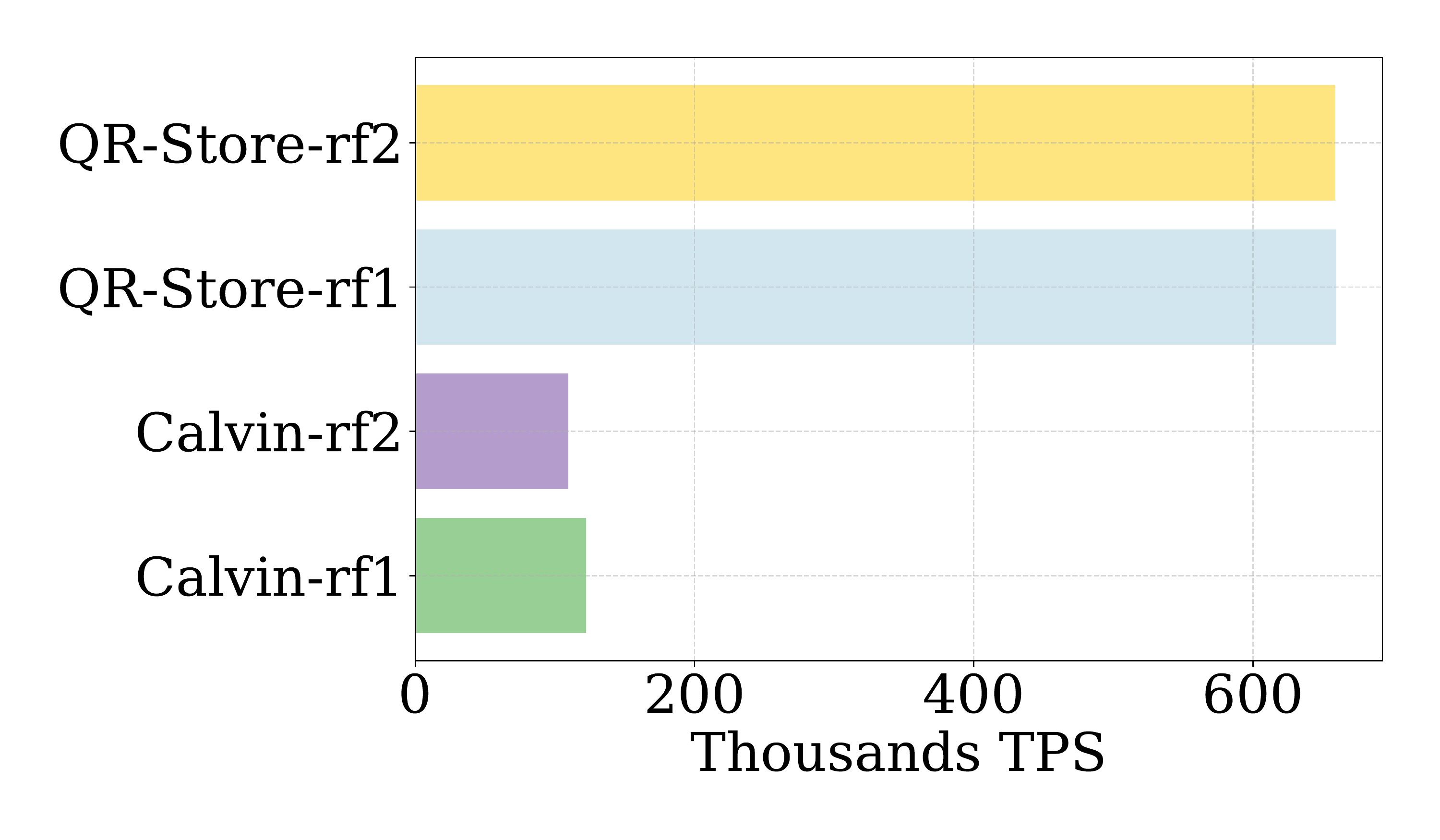}
	\centering
	\vspace{-3mm}
	\caption{Comparison with Calvin}
	\label{fig:eval-calvin}
	\vspace{-4mm}
\end{figure} 

\textbf{Comparison with Calvin}
In this set of experiments, we want compare \quecc{}'s performance to \calvin{}'s \cite{thomson_calvin_2012}. We implemented \calvin{}'s approach to replication which uses Paxos via Zookeeper. \quecc{}'s approach uses an integrated replication protocol (Section \ref{sec:qrep-impl}). Hence, \quecc{}'s implementation of the replication layer eliminates the overhead of a replication middleware. We use four server nodes per cluster and enable compression for \calvin{} replicated messages. We use a highly skewed workload for the workloads with $\theta = 0.9$, $10$ operations per transaction, $50\%$ MPT, update ratios, and force each transaction to access two partitions. 

As shown in Figure \ref{fig:eval-calvin}, \quecc{}'s configurations (denoted as \qsone{} and \qstwo{}) outperforms \calvin{}'s configurations (denoted as \calvinone{} and \calvintwo{}) by up to $6\times$. Multiple factors are contributing to this improvement. The first one is the use of the queue-oriented speculative transaction processing model, which is more efficient than \calvin{}'s transaction execution model. The second factor is the use of the integrated replication implementation as opposed to Zookeeper, which introduces significant replication processing overhead. Using the integrated approach, \quecc{} introduced no more than $8\%$ performance overhead compared to the \qstore{} within a four-node cluster configuration.

%


\begin{figure}[!t]
	\includegraphics[width=\plotwidth{0.82}{0.6}\columnwidth]{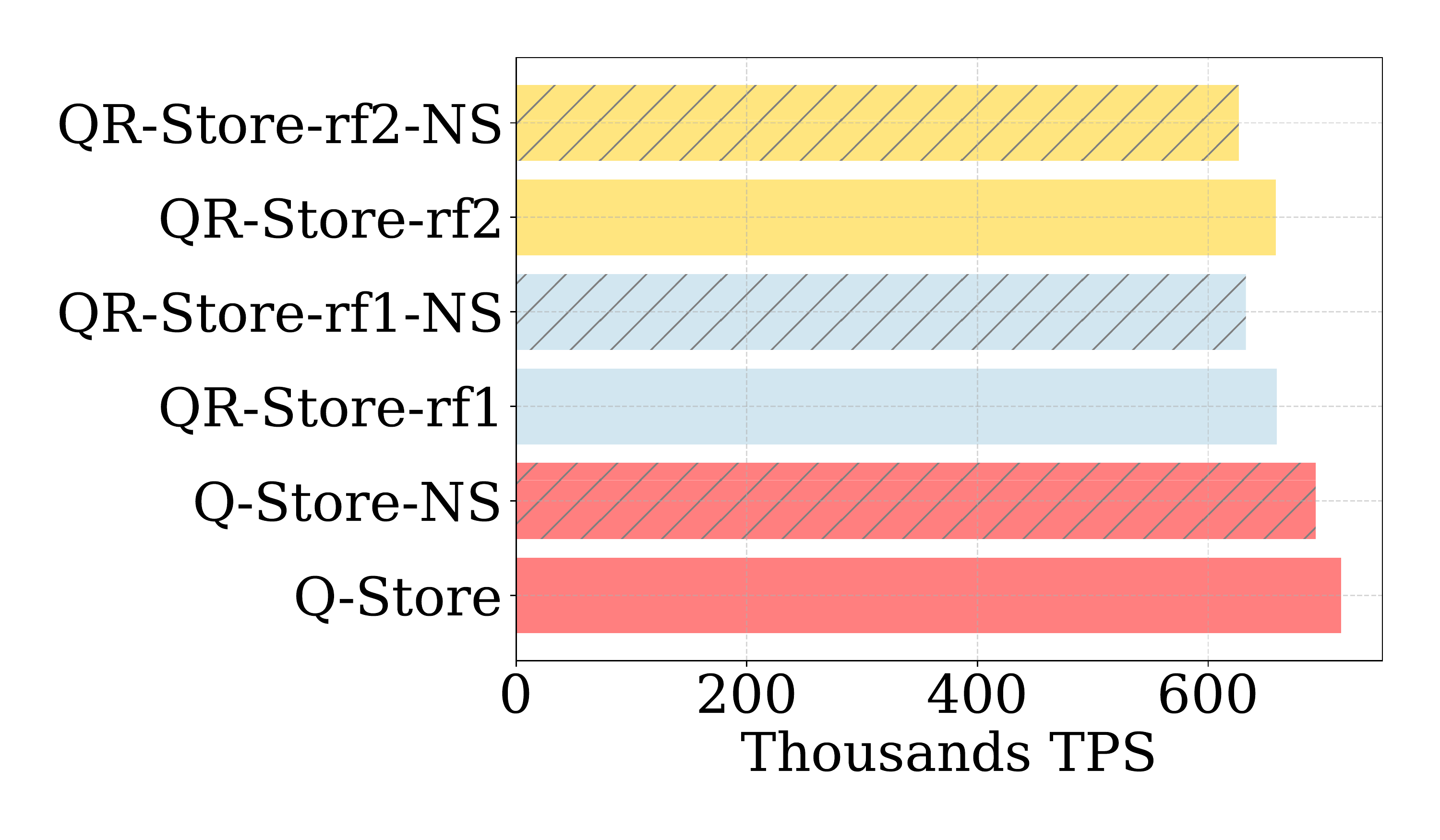}
	\centering
	\vspace{-3mm}
	\caption{Node granularity vs. fine granularity synchronization}
	\label{fig:eval-nodesync}
	\vspace{-4mm}
\end{figure} 

\textbf{Impact of node-granularity synchronization}
One of the key optimizations that we introduced in our current prototype is the granularity of synchronization in \quecc{}. In our previous work \cite{qadah_q-store_2020}, we adopted a node-level synchronization protocol that runs after processing a batch, which synchronizes all worker threads before they start working on the next batch. Using this approach simplified our prototype implementation and allowed us to avoid locking shared data structures. However, it also introduced unnecessary idle time periods where worker threads can perform useful work for the next batch. 

In our current prototype, we designed and implemented a fine-grained synchronization protocol that increases the concurrency of planning and execution phases. With our queue-oriented transaction processing paradigm, a node-level partition is further partitioned by planning threads. Instead of waiting for every other node in the cluster before starting its planning phase, it starts the planning phase, and it only waits for \ackmsg{} messages for the remote EQs that it planned before delivering the EQs for the new batch. Hence, this approach effectively implements a fine-grained synchronization protocol at the thread level instead of the node level.   

Figure \ref{fig:eval-nodesync}, shows a comparison of node-level synchronization and thread-level synchronization with various configurations of \quecc{}. We use four server nodes per cluster and the default workload parameters. The configurations that use node level synchronization are denoted with {\small \sf -NS} suffix. The thread-level synchronization technique provides up to $5\%$ improvement in throughput performance and up to $14\%$ in latency reduction.

\begin{figure}[!t]
	\includegraphics[width=\plotwidth{0.82}{0.6}\columnwidth]{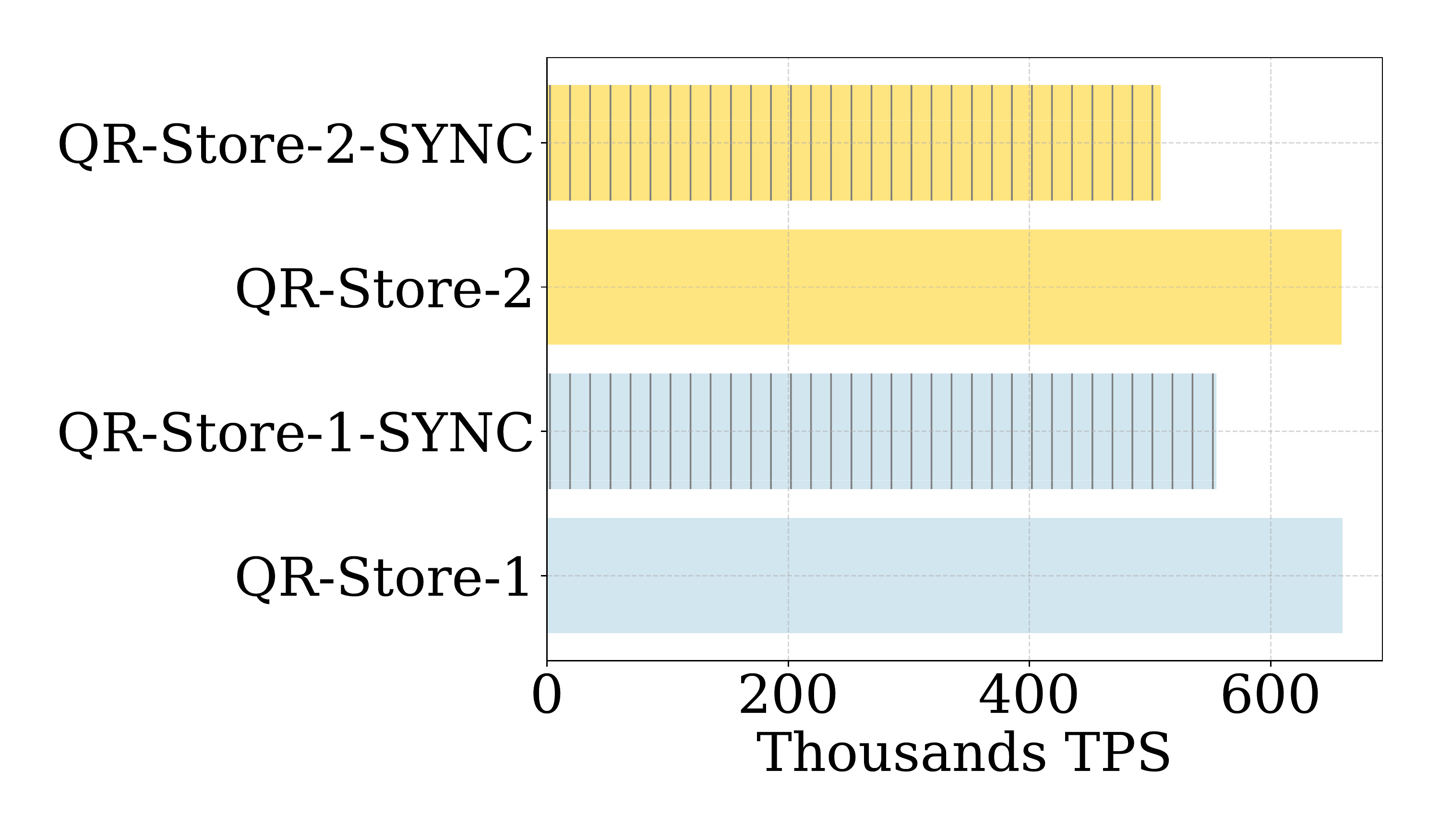}
	\centering
	\vspace{-3mm}
	\caption{Speculative replication versus synchronous replication}
	\label{fig:eval-svs}
	\vspace{-4mm}
\end{figure}

\textbf{Speculative replication vs. synchronous} 
Another key technique in \quecc{} is the concept of speculative replication. The basic idea is that instead of waiting for the replication to complete before starting the execution phase of a batch, \quecc{} speculates that the replication is expected to succeed and starts the execution phase without waiting. However, before starting the commit stage, the system waits for {\em acknowledgments} confirming the success of the replication requests. In Figure \ref{fig:eval-svs}, we show a comparison between the two techniques (the synchronous configuration is denoted with a {\small \sf -SYNC} suffix). The speculative replication technique improves the performance by up to $30\%$ with a four-node cluster and the default workload parameters. 

\begin{figure}[!t]
	\includegraphics[width=\plotwidth{0.7}{0.45}\columnwidth]{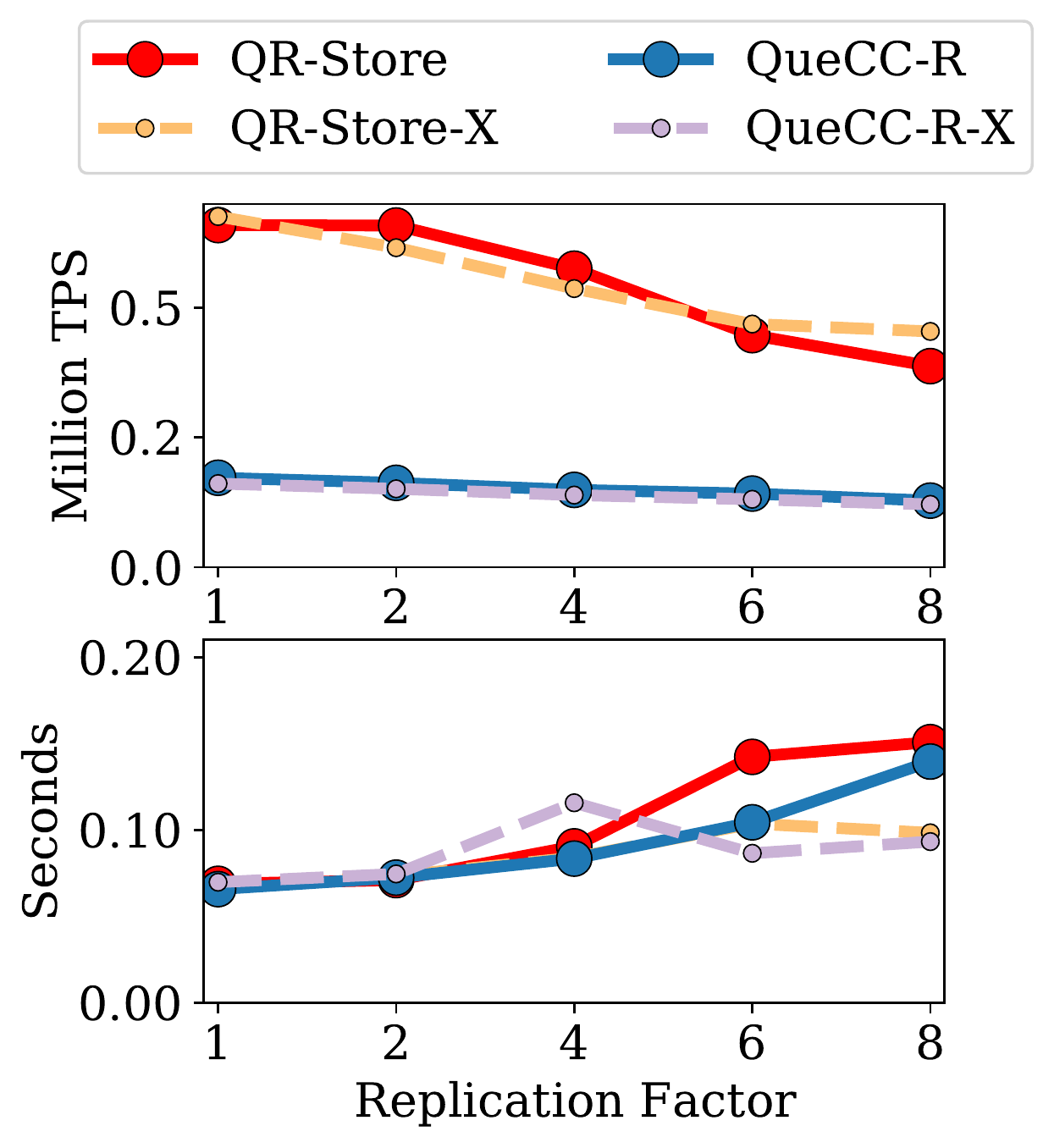}
	\centering
	\vspace{-3mm}
	\caption{Impact of the replication factor and replication compression. Dashed lines represent configurations with replication compression enabled. Four nodes per cluster for \quecc{}}
	\label{fig:eval-rfvar}
	\vspace{-4mm}
\end{figure} 

\textbf{Impact of the replication factor} 
The replication factor dictates the number of replicas of the database instance. The leader set of server nodes perform a proportional amount of work to the number of configured replicas. We perform a set of experiments involving four nodes per cluster and a fully replicated configuration. The fully replicated configuration implements the third case study described in Section \ref{sec:case-studies-qc}. In Figure \ref{fig:eval-rfvar}, \qsonenode{} is a non-partitioned (i.e., single-node database instance) and replicated configuration of \quecc{} while \qsfournode{} is the partitioned and replicated configuration with four-node per cluster. We increase the replication factor from $1$ up to $8$. As we can see, the overhead of replication beyond a replication factor of four becomes significant. It reduces the performance by up to $41\%$ and $26\%$ as we increase the replication factor to $8$ for \qsfournode{} and \qsonenode{}, respectively. The large drop is due to the increased demand for network resources as the number of replication messages increases proportionally to the replication factor.

\textbf{Impact of using compression for replication} 
Intuitively, using compression reduces the number of bytes that go over the network for replication messages by the leader set of nodes in both \quecc{} and \qsonenode{}. However, it increases the CPU computation requirements on the leader set of nodes as more CPU cycles are needed to perform the compression. Hence, compression is not a silver bullet and is not always beneficial. We conduct experiments to demonstrate that. Figure \ref{fig:eval-rfvar} shows that compression can be beneficial when the replication factor is high (e.g., at $6$ or $8$). The $99^{th}$ percentile latency improves by up to $35\%$. Compression is also beneficial for \calvin{} because it uses Zookpeer as the replication layer. In our experiments, \calvin{} latency improved by up to $53\%$. This result agrees with our micro-benchmark result shown in Figure \ref{fig:zklatency}. However, it increases the work on the leader nodes in all other cases, which negates its benefits.



%
%
%
%

\section{Related Work}

In this section, we discuss relevant work from the literature. In this paper, we addressed an important challenge of high-performance replication in distributed deterministic transaction processing systems. However, the database replication techniques have been studied since several decades ago \cite{gray_dangers_1996}. These techniques are mainly studied with respect to two dimensions. The first dimension is whether to allow transactions to update at any replica or designate a primary copy replica. The second dimension is when to synchronize replicas and whether we should do that synchronization eagerly or lazily. Furthermore, traditional database replication techniques reuse non-deterministic transaction processing protocols \cite{sadoghi_transaction_2019} e.g., 2PL, 2PC \cite{gray_transaction_1992} and OCC \cite{kung_optimistic_1981}. The reader is referred to existing literature that cover the traditional database replication techniques very well (e.g., \cite{kemme_dont_2000,kemme_new_2000, kemme_database_2010, kemme_database_2010-1, ozsu_principles_2020}) for more details. 
Compared to traditional database replication techniques, the techniques proposed in this paper are deterministic, speculative, and adopt the queue-oriented transaction processing paradigm \cite{qadah_queue-oriented_2019}. Hence, this paper explores a new research territory. 

\textbf{Replication Frameworks} Wiesmann et al. \cite{wiesmann_understanding_2000} proposed a general replication framework to study replication techniques developed by the database systems research community and the distributed systems research community. However, it does not address the design choices made by DTP systems. Our proposed general framework complements their replication framework by focusing on replications in DTP systems. 

\textbf{Deterministic transaction processing protocols} 
Deterministic transaction processing approaches are shown to process transactions more efficiently when compared to non-deterministic approaches. Recently, there have been many proposals for deterministic transaction processing protocols. These systems can be {\em centralized} (e.g., \cite{qadah_quecc:_2018, faleiro_high_2017, faleiro_lazy_2014, yao_exploiting_2016, faleiro_rethinking_2015}) or {\em distributed} (e.g., \cite{kallman_h-store_2008, thomson_calvin_2012, qadah_q-store_2020, lu_aria_2020}). The main focus of these proposals is on the concurrency control aspects of processing transactions. In contrast, this paper goes beyond existing deterministic transaction processing systems by addressing the replication challenge in a systematic way and build on queue-oriented principles, which allows efficient transaction processing and database replication. 

\textbf{Replication and Consensus Protocols}
In the distributed systems community, consensus and replication have received a great deal of attention (e.g., \cite{oki_viewstamped_1988, lamport_part-time_1998,lamport_paxos_2002,ongaro_search_2014}). Such work focused on state-machine replication and aimed to achieve linearizability. In contrast, this paper is concerned with strict serializability guarantees of transaction processing on distributed, partitioned, and replicated databases. 


\section{Conclusion}
In this paper, we propose a generalized framework for designing replication schemes for distributed DTP systems. Using the framework, we study three cases from the literature and discuss how replication can be reasoned about. We propose a novel queue-oriented speculative replication technique and describe how it is supported in \quecc{}. Finally, we perform an extensive evaluation of several configurations of \quecc{} and demonstrate efficient replicated transaction processing that can reach up to $1.9$ million replicated transactions per second in under $200$ milliseconds and a replication overhead of $8\%-25\%$ compared to non-replicated configurations.

\printbibliography

@article{abadi_overview_2018,
  title = {An {{Overview}} of {{Deterministic Database Systems}}},
  author = {Abadi, Daniel J. and Faleiro, Jose M.},
  date = {2018-08},
  journaltitle = {Commun. ACM},
  volume = {61},
  pages = {78--88},
  issn = {0001-0782},
  doi = {10.1145/3181853},
  url = {http://doi.acm.org/10.1145/3181853},
  urldate = {2018-10-18},
  number = {9}
}

@online{apache_kafka_website,
  title = {Apache {{Kafka}}},
  url = {https://kafka.apache.org/},
  urldate = {2021-05-05},
  langid = {english},
  organization = {{Apache Kafka}}
}

@online{apache_zookeeper_website,
  title = {Apache {{ZooKeeper}}},
  url = {https://zookeeper.apache.org/},
  urldate = {2021-05-05}
}

@online{etcd_website,
  title = {Etcd},
  url = {https://etcd.io/},
  urldate = {2021-05-05},
  langid = {english},
  organization = {{etcd}}
}

@article{faleiro_high_2017,
  title = {High {{Performance Transactions}} via {{Early Write Visibility}}},
  author = {Faleiro, Jose M. and Abadi, Daniel J. and Hellerstein, Joseph M.},
  date = {2017-01},
  journaltitle = {Proc. VLDB Endow.},
  volume = {10},
  pages = {613--624},
  issn = {2150-8097},
  doi = {10.14778/3055540.3055553},
  url = {http://dx.doi.org/10.14778/3055540.3055553},
  number = {5}
}

@inproceedings{faleiro_lazy_2014,
  title = {Lazy {{Evaluation}} of {{Transactions}} in {{Database Systems}}},
  booktitle = {Proc. {{SIGMOD}}},
  author = {Faleiro, Jose M. and Thomson, Alexander and Abadi, Daniel J.},
  date = {2014},
  pages = {15--26},
  publisher = {{ACM}},
  doi = {10.1145/2588555.2610529},
  url = {http://dx.doi.org/10.1145/2588555.2610529},
  isbn = {978-1-4503-2376-5}
}

@article{faleiro_rethinking_2015,
  title = {Rethinking {{Serializable Multiversion Concurrency Control}}},
  author = {Faleiro, Jose M. and Abadi, Daniel J.},
  date = {2015-07},
  journaltitle = {Proc. VLDB Endow.},
  volume = {8},
  pages = {1190--1201},
  issn = {2150-8097},
  doi = {10.14778/2809974.2809981},
  url = {http://dx.doi.org/10.14778/2809974.2809981},
  number = {11}
}

@article{gray_dangers_1996,
  title = {The {{Dangers}} of {{Replication}} and a {{Solution}}},
  author = {Gray, Jim and Helland, Pat and O'Neil, Patrick and Shasha, Dennis},
  date = {1996-06},
  journaltitle = {SIGMOD Rec.},
  volume = {25},
  pages = {173--182},
  issn = {0163-5808},
  doi = {10.1145/235968.233330},
  url = {http://dx.doi.org/10.1145/235968.233330},
  number = {2}
}

@book{gray_transaction_1992,
  title = {Transaction {{Processing}}: {{Concepts}} and {{Techniques}}},
  author = {Gray, Jim and Reuter, Andreas},
  date = {1992},
  edition = {1st},
  publisher = {{Morgan Kaufmann Publishers Inc.}},
  location = {{San Francisco, CA, USA}},
  url = {http://portal.acm.org/citation.cfm?id=573304},
  isbn = {1-55860-190-2}
}

@inproceedings{junqueira_zab_2011,
  title = {Zab: {{High}}-Performance Broadcast for Primary-Backup Systems},
  shorttitle = {Zab},
  booktitle = {2011 {{IEEE}}/{{IFIP}} 41st {{International Conference}} on {{Dependable Systems Networks}} ({{DSN}})},
  author = {Junqueira, Flavio P. and Reed, Benjamin C. and Serafini, Marco},
  date = {2011-06},
  pages = {245--256},
  issn = {2158-3927},
  doi = {10.1109/DSN.2011.5958223},
  eventtitle = {2011 {{IEEE}}/{{IFIP}} 41st {{International Conference}} on {{Dependable Systems Networks}} ({{DSN}})}
}

@article{kallman_h-store_2008,
  title = {H-Store: {{A High}}-Performance, {{Distributed Main Memory Transaction Processing System}}},
  author = {Kallman, Robert and Kimura, Hideaki and Natkins, Jonathan and Pavlo, Andrew and Rasin, Alexander and Zdonik, Stanley and Jones, Evan P. C. and Madden, Samuel and Stonebraker, Michael and Zhang, Yang and Hugg, John and Abadi, Daniel J.},
  date = {2008-08},
  journaltitle = {Proc. VLDB Endow.},
  volume = {1},
  pages = {1496--1499},
  issn = {2150-8097},
  doi = {10.14778/1454159.1454211},
  url = {http://dx.doi.org/10.14778/1454159.1454211},
  number = {2}
}

@book{kemme_database_2010,
  title = {Database Replication},
  author = {Kemme, Bettina and Jiménez-Peris, Ricardo and Patiño-Martínez, Marta},
  date = {2010},
  publisher = {{Morgan \& Claypool Publishers}},
  doi = {10.2200/S00296ED1V01Y201008DTM007},
  url = {https://doi.org/10.2200/S00296ED1V01Y201008DTM007},
  bibsource = {dblp computer science bibliography, https://dblp.org},
  series = {Synthesis Lectures on Data Management}
}

@inproceedings{kemme_database_2010-1,
  title = {Database Replication: {{A}} Tutorial},
  booktitle = {Replication: {{Theory}} and Practice},
  author = {Kemme, Bettina and Jiménez-Peris, Ricardo and Patiño-Martínez, Marta and Alonso, Gustavo},
  editor = {Charron-Bost, Bernadette and Pedone, Fernando and Schiper, André},
  date = {2010},
  volume = {5959},
  pages = {219--252},
  publisher = {{Springer}},
  doi = {10.1007/978-3-642-11294-2\\_12},
  bibsource = {dblp computer science bibliography, https://dblp.org},
  series = {Lecture Notes in Computer Science}
}

@inproceedings{kemme_dont_2000,
  title = {Don'{{T Be Lazy}}, {{Be Consistent}}: {{Postgres}}-{{R}}, {{A New Way}} to {{Implement Database Replication}}},
  booktitle = {Proc. {{VLDB}}},
  author = {Kemme, Bettina and Alonso, Gustavo},
  date = {2000},
  pages = {134--143},
  publisher = {{Morgan Kaufmann Publishers Inc.}},
  url = {http://dl.acm.org/citation.cfm?id=645926.671855}
}

@article{kemme_new_2000,
  title = {A New Approach to Developing and Implementing Eager Database Replication Protocols},
  author = {Kemme, Bettina and Alonso, Gustavo},
  date = {2000-09-01},
  journaltitle = {ACM Transactions on Database Systems},
  shortjournal = {ACM Trans. Database Syst.},
  volume = {25},
  pages = {333--379},
  issn = {0362-5915},
  doi = {10.1145/363951.363955},
  url = {http://doi.org/10.1145/363951.363955},
  urldate = {2021-07-07},
  number = {3}
}

@article{kung_optimistic_1981,
  title = {On {{Optimistic Methods}} for {{Concurrency Control}}},
  author = {Kung, H. T. and Robinson, John T.},
  date = {1981-06},
  journaltitle = {ACM Trans. Database Syst.},
  volume = {6},
  pages = {213--226},
  issn = {0362-5915},
  doi = {10.1145/319566.319567},
  url = {http://dx.doi.org/10.1145/319566.319567},
  number = {2}
}

@article{lamport_part-time_1998,
  title = {The Part-Time Parliament},
  author = {Lamport, Leslie},
  date = {1998-05-01},
  journaltitle = {ACM Transactions on Computer Systems},
  shortjournal = {ACM Trans. Comput. Syst.},
  volume = {16},
  pages = {133--169},
  issn = {0734-2071},
  doi = {10.1145/279227.279229},
  url = {http://doi.org/10.1145/279227.279229},
  urldate = {2021-05-05},
  number = {2}
}

@inproceedings{lamport_paxos_2002,
  title = {Paxos Made Simple, Fast, and Byzantine},
  booktitle = {Procedings of the 6th International Conference on Principles of Distributed Systems. {{OPODIS}} 2002, Reims, France, December 11-13, 2002},
  author = {Lamport, Leslie},
  editor = {Bui, Alain and Fouchal, Hacène},
  date = {2002},
  volume = {3},
  pages = {7--9},
  publisher = {{Suger, Saint-Denis, rue Catulienne, France}},
  bibsource = {dblp computer science bibliography, https://dblp.org},
  series = {Studia Informatica Universalis}
}

@article{lu_aria_2020,
  title = {Aria: A Fast and Practical Deterministic {{OLTP}} Database},
  shorttitle = {Aria},
  author = {Lu, Yi and Yu, Xiangyao and Cao, Lei and Madden, Samuel},
  date = {2020-07-01},
  journaltitle = {Proceedings of the VLDB Endowment},
  shortjournal = {Proc. VLDB Endow.},
  volume = {13},
  pages = {2047--2060},
  issn = {2150-8097},
  doi = {10.14778/3407790.3407808},
  url = {https://doi.org/10.14778/3407790.3407808},
  urldate = {2021-07-10},
  number = {12}
}

@article{lu_star_2019,
  title = {{{STAR}}: {{Scaling Transactions}} through {{Asymmetric Replication}}},
  author = {Lu, Yi and Yu, Xiangyao and Madden, Samuel},
  date = {2019},
  journaltitle = {PVLDB},
  volume = {12},
  pages = {1316--1329},
  url = {http://www.vldb.org/pvldb/vol12/p1316-lu.pdf},
  number = {11}
}

@inproceedings{oki_viewstamped_1988,
  title = {Viewstamped {{Replication}}: {{A New Primary Copy Method}} to {{Support Highly}}-{{Available Distributed Systems}}},
  shorttitle = {Viewstamped {{Replication}}},
  booktitle = {Proceedings of the Seventh Annual {{ACM Symposium}} on {{Principles}} of Distributed Computing},
  author = {Oki, Brian M. and Liskov, Barbara H.},
  date = {1988-01-01},
  pages = {8--17},
  publisher = {{Association for Computing Machinery}},
  location = {{New York, NY, USA}},
  doi = {10.1145/62546.62549},
  url = {http://doi.org/10.1145/62546.62549},
  urldate = {2021-05-05},
  isbn = {978-0-89791-277-8},
  series = {{{PODC}} '88}
}

@inproceedings{ongaro_search_2014,
  title = {In Search of an Understandable Consensus Algorithm},
  booktitle = {Proceedings of the 2014 {{USENIX}} Conference on {{USENIX Annual Technical Conference}}},
  author = {Ongaro, Diego and Ousterhout, John},
  date = {2014-06-19},
  pages = {305--320},
  publisher = {{USENIX Association}},
  location = {{USA}},
  isbn = {978-1-931971-10-2},
  series = {{{USENIX ATC}}'14}
}

@book{ozsu_principles_2020,
  title = {Principles of {{Distributed Database Systems}}},
  author = {Özsu, M. Tamer and Valduriez, Patrick},
  date = {2020},
  edition = {4},
  publisher = {{Springer International Publishing}},
  doi = {10.1007/978-3-030-26253-2},
  url = {https://www.springer.com/gp/book/9783030262525},
  urldate = {2020-12-06},
  isbn = {978-3-030-26252-5},
  langid = {english}
}

@inproceedings{qadah_q-store_2020,
  title = {Q-Store: {{Distributed}}, Multi-Partition Transactions via Queue-Oriented Execution and Communication},
  booktitle = {Proceedings of the 23rd International Conference on Extending Database Technology, {{EDBT}} 2020, Copenhagen, Denmark, March 30 - April 02, 2020},
  author = {Qadah, Thamir and Gupta, Suyash and Sadoghi, Mohammad},
  editor = {Bonifati, Angela and Zhou, Yongluan and Salles, Marcos Antonio Vaz and Böhm, Alexander and Olteanu, Dan and Fletcher, George H. L. and Khan, Arijit and Yang, Bin},
  date = {2020},
  pages = {73--84},
  publisher = {{OpenProceedings.org}},
  doi = {10.5441/002/edbt.2020.08},
  url = {https://doi.org/10.5441/002/edbt.2020.08},
  bibsource = {dblp computer science bibliography, https://dblp.org}
}

@inproceedings{qadah_quecc:_2018,
  title = {{{QueCC}}: {{A Queue}}-Oriented, {{Control}}-Free {{Concurrency Architecture}}},
  shorttitle = {{{QueCC}}},
  booktitle = {Proceedings of the 19th {{International Middleware Conference}}},
  author = {Qadah, Thamir M. and Sadoghi, Mohammad},
  date = {2018},
  pages = {13--25},
  publisher = {{ACM}},
  location = {{New York, NY, USA}},
  doi = {10.1145/3274808.3274810},
  url = {http://doi.acm.org/10.1145/3274808.3274810},
  urldate = {2019-10-02},
  isbn = {978-1-4503-5702-9},
  series = {Middleware '18},
  venue = {Rennes, France}
}

@inproceedings{qadah_queue-oriented_2019,
  title = {A Queue-Oriented Transaction Processing Paradigm},
  booktitle = {Proceedings of the 20th International Middleware Conference Doctoral Symposium, Middleware 2019, Davis, {{CA}}, {{USA}}, December 09-13, 2019},
  author = {Qadah, Thamir M.},
  editor = {Nawab, Faisal and Riviere, Etienne},
  date = {2019},
  pages = {26--30},
  publisher = {{ACM}},
  doi = {10.1145/3366624.3368163},
  url = {https://doi.org/10.1145/3366624.3368163},
  bibsource = {dblp computer science bibliography, https://dblp.org}
}

@article{sadoghi_transaction_2019,
  title = {Transaction {{Processing}} on {{Modern Hardware}}},
  author = {Sadoghi, Mohammad and Blanas, Spyros},
  date = {2019-03-08},
  journaltitle = {Synthesis Lectures on Data Management},
  shortjournal = {Synthesis Lectures on Data Management},
  volume = {14},
  pages = {1--138},
  issn = {2153-5418},
  doi = {10.2200/S00896ED1V01Y201901DTM058},
  url = {https://www.morganclaypool.com/doi/10.2200/S00896ED1V01Y201901DTM058},
  urldate = {2019-12-03},
  number = {2}
}

@software{snappy_2021,
  title = {Google/Snappy},
  date = {2021-05-08T12:49:12Z},
  origdate = {2014-03-03T21:58:09Z},
  url = {https://github.com/google/snappy},
  urldate = {2021-05-08},
  organization = {{Google}}
}

@inproceedings{thomson_calvin_2012,
  title = {Calvin: {{Fast Distributed Transactions}} for {{Partitioned Database Systems}}},
  booktitle = {Proc. {{SIGMOD}}},
  author = {Thomson, Alexander and Diamond, Thaddeus and Weng, Shu C. and Ren, Kun and Shao, Philip and Abadi, Daniel J.},
  date = {2012},
  pages = {1--12},
  publisher = {{ACM}},
  doi = {10.1145/2213836.2213838},
  url = {http://dx.doi.org/10.1145/2213836.2213838},
  isbn = {978-1-4503-1247-9}
}

@inproceedings{wiesmann_understanding_2000,
  title = {Understanding Replication in Databases and Distributed Systems},
  booktitle = {Proceedings 20th {{IEEE International Conference}} on {{Distributed Computing Systems}}},
  author = {Wiesmann, M. and Pedone, F. and Schiper, A. and Kemme, B. and Alonso, G.},
  date = {2000-04},
  pages = {464--474},
  issn = {1063-6927},
  doi = {10.1109/ICDCS.2000.840959},
  eventtitle = {Proceedings 20th {{IEEE International Conference}} on {{Distributed Computing Systems}}}
}

@article{yao_exploiting_2016,
  title = {Exploiting {{Single}}-{{Threaded Model}} in {{Multi}}-{{Core In}}-{{Memory Systems}}},
  author = {Yao, C. and Agrawal, D. and Chen, G. and Lin, Q. and Ooi, B. C. and Wong, W. F. and Zhang, M.},
  date = {2016},
  journaltitle = {IEEE TKDE},
  volume = {28},
  pages = {2635--2650},
  doi = {10.1109/TKDE.2016.2578319},
  url = {http://dx.doi.org/10.1109/TKDE.2016.2578319},
  number = {10}
}

@inproceedings{zamanian_chiller_2020,
  title = {Chiller: {{Contention}}-Centric {{Transaction Execution}} and {{Data Partitioning}} for {{Modern Networks}}},
  shorttitle = {Chiller},
  booktitle = {Proceedings of the 2020 {{ACM SIGMOD International Conference}} on {{Management}} of {{Data}}},
  author = {Zamanian, Erfan and Shun, Julian and Binnig, Carsten and Kraska, Tim},
  date = {2020-06-11},
  pages = {511--526},
  publisher = {{Association for Computing Machinery}},
  location = {{Portland, OR, USA}},
  doi = {10.1145/3318464.3389724},
  url = {http://doi.org/10.1145/3318464.3389724},
  urldate = {2020-06-20},
  isbn = {978-1-4503-6735-6},
  series = {{{SIGMOD}} '20}
}






\end{document}